\documentclass[a4paper,11pt]{article}

\usepackage{jheppub}
\usepackage{amsthm}
\usepackage{physics}
\usepackage{empheq}
\usepackage{color}
\usepackage{aliascnt}
\usepackage{cleveref}

\usepackage{tikz}
\usetikzlibrary{fit,backgrounds,calc,positioning}
\tikzset{
  vtx/.style={circle,draw,fill=black,inner sep=1.6pt},
  vlabel/.style={above=2pt,font=\small},
  hred/.style={draw=red!70!black,fill=red!12,thick,rounded corners=8pt},
  hblue/.style={draw=blue!70!black,fill=blue!10,thick,rounded corners=8pt},
  hgreen/.style={draw=green!50!black,fill=green!10,thick,rounded corners=8pt},
  hedge/.style={font=\small}
}

\makeatletter
\gdef\@fpheader{}
\makeatother

\Crefname{equation}{Eq.}{Eqs.}
\Crefname{appendix}{Appendix}{Appendices}
\Crefname{section}{Section}{Sections}
\Crefname{theorem}{Theorem}{Theorems}
\Crefname{corollary}{Corollary}{Corollaries}

\newtheorem{theorem}{Theorem}[section]
\newaliascnt{corollary}{theorem}
\newtheorem{corollary}[corollary]{Corollary}
\aliascntresetthe{corollary}

\numberwithin{equation}{section}

\newcommand{\del}{\partial}
\newcommand{\inp}[2]{\langle #1, #2 \rangle}

\DeclareMathOperator{\supp}{supp}

\begin{document}

\title{Note on Logical Gates by Gauge Field Formalism of Quantum Error Correction}
\author{Junichi Haruna}
\emailAdd{j.haruna1111@gmail.com}
\affiliation{Graduate School of Informatics, Kyoto University, Kyoto 606-8501, Japan}
\affiliation{Center for Quantum Information and Quantum Biology, University of Osaka, Toyonaka, Osaka 560-0043, Japan}
\date{\today}

\abstract{
    The gauge field formalism, or operator-valued cochain formalism, has recently emerged as a powerful framework for describing quantum Calderbank--Shor--Steane (CSS) codes.
    In this work, we extend this framework to construct a broad class of logical gates for general CSS codes, including the $S$, Hadamard, $T$, and (multi)controlled-$Z$ gates, under the condition where fault-tolerance or circuit-depth optimality is not necessarily imposed.
    We show that these logical gates can be expressed as exponentials of polynomial functions of the electric and magnetic gauge fields, which allows us to derive explicit decompositions into physical gates.
    We further prove that their logical action depends only on the (co)homology classes of the corresponding logical qubits, establishing consistency as logical operations.
    Our results provide a systematic method for formulating logical gates for general CSS codes, offering new insights into the interplay between quantum error correction, algebraic topology, and quantum field theory.
}


\maketitle

\section{Introduction}

Quantum computers have advanced rapidly in recent years, promising computational advantages over classical machines for specific tasks~\cite{Shor:1994jg,Grover:1996rk}.
However, their practical realization faces significant challenges because of the fragility of quantum information and its susceptibility to noise.
Quantum error correction (QEC)~\cite{Shor:1995hbe} provides the essential framework for mitigating such noise by encoding logical qubits into protected subspaces.

Among the many QEC architectures, topological codes such as the toric code~\cite{Kitaev:1997wr} and color codes~\cite{Bombin:2006cd} have played a central role.
Their geometric locality and high noise thresholds~\cite{Dennis_2002} make them particularly compatible with superconducting-qubit platforms~\cite{barends2014superconducting,Takita_2017,acharya2024quantum}.
In parallel, recent progress in hardware based on trapped ions~\cite{Wineland:1997mg,Kielpinski:2002wbd,Ransford:2025ksn}, neutral atoms~\cite{Bluvstein:2023zmt,Rodriguez:2024bhh}, and photonic systems~\cite{Zhong:2020bnd,Alexander:2024vys} has stimulated significant interest in quantum low-density parity-check (qLDPC) codes~\cite{poulin2008iterativedecodingsparsequantum,Tillich_2014,Panteleev:2019upy}.
These codes offer constant-rate encoding and favorable distance scaling, suggesting a path toward large-scale fault-tolerant architectures.

While code design is essential, implementing logical gates on encoded qubits is equally critical.
Existing approaches such as transversal gates~\cite{Zeng:2007uho,Eastin:2009tem,Bravyi:2012rnv}, lattice surgery~\cite{Horsman_2012,Fowler_2012,Litinski_2019}, code deformation~\cite{Bombin:2009tbs}, and magic-state distillation~\cite{Knill:2004ctr,Bravyi:2004isx} provide powerful tools but typically apply only to specific code families or rely on particular geometric structures.
A general and systematic framework for constructing logical gates for general QEC codes remains lacking.

The gauge field formalism, or operator-valued cochain formalism~\cite{Barkeshli:2022edm,Zhu:2023xfg,Lin:2024uhb,Golowich:2024ogv,breuckmann2024cupsgatesicohomology,hsin2024classifyinglogicalgatesquantum,zhu2025topologicaltheoryqldpcnonclifford,Kobayashi:2025cfh}, has recently emerged as a promising route toward such a framework.
This formalism leverages the correspondence between the quantum Calderbank--Shor--Steane (CSS) codes~\cite{Calderbank_1996,Steane:1996ghp} and chain complexes, known as the CSS--homology correspondence~\cite{Kitaev1997quantum,Bombin:2006cd}, identifying qubits and stabilizers with 1-chains and 0-/2-chains of a chain complex, respectively.
Pauli operators are expressed as exponentials of operator-valued cochains, which are interpreted as electric and magnetic gauge fields, and logical Pauli operators appear as Wilson loops associated with homology and cohomology classes~\cite{Barkeshli:2022edm}.
This perspective provides both algebraic clarity and a geometric interpretation reminiscent of $\mathbb{Z}_2$ lattice gauge theory~\cite{Kogut:1974ag,Kogut:1979wt} in quantum field theory.

Recent works have used this formalism to construct logical gates in various settings.
Refs.~\cite{Zhu:2023xfg,Lin:2024uhb,Golowich:2024ogv,zhu2025topologicaltheoryqldpcnonclifford} realized transversal non-Clifford gates in qLDPC codes defined on chain complexes with special geometric structures.
Ref.~\cite{breuckmann2024cupsgatesicohomology} showed that cohomology invariants yield diagonal logical gates and constructed logical multicontrolled-$Z$ gates for chain complexes with a cup-product structure.
Ref.~\cite{hsin2024classifyinglogicalgatesquantum} extended this formalism to higher-form gauge fields and demonstrated that higher cohomology invariants such as Steenrod squares and higher Pontryagin powers generate families of Clifford and non-Clifford logical gates on closed or open manifolds.

These developments highlight the versatility of the gauge field viewpoint, but current constructions typically rely on geometric embeddings or specific product complexes and primarily address diagonal gates; a unified treatment of nondiagonal gates such as Hadamard is still missing.
Moreover, a systematic procedure to decompose logical gates into physical operations within this formalism has not been fully established.

In this work, we provide one way to construct a broad class of logical gates using this formalism.
Our construction does not require any special manifold or product structure for the underlying chain complex, and so can be applied for general CSS codes.
We should note here that our focus is not on fault tolerance such as constant depth or locality, but rather on establishing an algebraically transparent foundation upon which such considerations may later be incorporated.

Using the electric and magnetic gauge fields associated with the underlying chain complex, we derive explicit physical decompositions for the logical $S$, Hadamard, $T$, $CZ$, and multicontrolled-$Z$ gates. These logical gates admit compact polynomial expressions in the gauge fields, and we show that their logical operation depends only on the (co)homology classes of their associated logical qubit
This property ensures that our physical gate decompositions are well defined at the logical level.
A summary of the constructed logical gates and their physical decompositions is presented in Table~\ref{tab:logical_gates_summary}.

The remainder of this paper is organized as follows.
\Cref{sec:Review} reviews QEC and the correspondence between CSS codes and chain complexes and introduces the gauge field formalism.
\Cref{sec:logical_gate_and_gauge_field_formalism} develops our method for constructing logical gates from physical ones and presents explicit examples including the $S$, Hadamard, $T$, and (multi)controlled-$Z$ gates.
We also show that these logical gates depend only on (co)homology classes on the code space.
\Cref{sec:conclusion} is dedicated to the conclusion and outlook.

\renewcommand{\arraystretch}{1.6}
\begin{table}[htbp]
    \centering
    \begin{tabular}{|c|c|c|}
        \hline
        Logical gate & Physical gate decomposition & Gauge field expression \\
        \hline
        $\overline{Z}(\gamma)$ & $\displaystyle \prod_k Z_{i_k}$ & $e^{i\pi a(\gamma)}$ \\
        $\overline{X}(\tilde{\gamma})$ & $\displaystyle \prod_k X_{\tilde{i}_k}$ & $e^{i \pi b(\tilde{\gamma})}$ \\
        $\overline{S}(\gamma)$ & $\displaystyle \prod_k S_{i_k} \prod_{k_1<k_2} CZ_{i_{k_1},i_{k_2}}$ & $e^{i \frac{\pi}{2}a(\gamma)^2}$ \\
        $\overline{H}(\gamma)$ & $\displaystyle e^{-\frac{i\pi}{4}}\overline{S}(\gamma)\cdot \prod_{k} H_{\tilde{i}_k}\cdot \overline{S}(\tilde{\gamma}) \cdot \prod_{k} H_{\tilde{i}_k} \cdot \overline{S}(\gamma)$& $e^{-i\frac{\pi}{4}} e^{i\frac{\pi}{2}a(\gamma)^2} e^{i\frac{\pi}{2}b(\tilde{\gamma})^2}e^{i\frac{\pi}{2}a(\gamma)^2}$ \\
        $\overline{CZ}(\gamma_1,\gamma_2)$ & $\displaystyle \prod_{k_1,k_2} CZ_{i^1_{k_1},i^2_{k_2}}$ & $e^{i \pi a(\gamma_1)a(\gamma_2)}$ \\
        $\overline{C^{m-1}Z}(\gamma_1,\ldots,\gamma_m)$ & $\displaystyle \prod_{k_1,\ldots,k_m} C^{m-1}Z_{i^1_{k_1},\ldots,i^m_{k_m}}$ & $e^{i \pi a(\gamma_1) \cdots a(\gamma_{m})}$\\
        $\overline{T}(\gamma)$ & $\displaystyle \prod_k T_{i_k} \prod_{k_1<k_2} CS_{i_{k_1},i_{k_2}}^\dagger \prod_{k_1<k_2<k_3} CCZ_{i_{k_1}i_{k_2},i_{k_3}}$ & $e^{i \frac{\pi}{4} (2 a(\gamma)^3 - 3 a(\gamma)^2 + 2 a(\gamma))}$\\
        \hline
    \end{tabular}
    \caption{
    Summary of logical gates constructed in this work, together with their physical gate decompositions and gauge field expressions.  
    The electric and magnetic gauge fields associated with Pauli $Z$ and $X$ operators are denoted by $a$ and $b$, respectively.  
    The symbols $\gamma$, $\tilde{\gamma}$, and $(\gamma_1,\ldots,\gamma_m)$ denote sets of physical qubits supporting the corresponding logical qubit, with indices $i_k$, $\tilde{i}_k$, and $i^j_{k_j}$ labeling the individual qubits.  
    Multicontrolled-$Z$ gates $C^{m-1}Z$ are defined such that repeated control or target qubits reduce to lower-control gates, and a control and target acting on the same qubit reduces to a single-qubit $Z$ operator (e.g. $C^3Z_{i,i,j,k}=C^2Z_{i,j,k}$ and $C^2Z_{i,i,i}=CZ_{i,i}=Z_i$).
    The controlled-$S$ gate $CS$ is defined similarly.
    }
    \label{tab:logical_gates_summary}
\end{table} 
\renewcommand{\arraystretch}{1.0}

\section{Preliminaries}
\label{sec:Review}

In this section, we briefly review the quantum error correction, focusing on the relation between the quantum CSS code and the chain complex.
After that, we review the gauge field formalism of the CSS code.

\subsection{Quantum error correction, CSS code, and chain complex}

Quantum error correction (QEC) is a method to protect quantum information from noise by encoding logical qubits into a larger number of physical qubits.  
The code space (or logical space)  $\mathcal{C}$ is defined as a subspace of the full Hilbert space of $n$ qubits, $\mathcal{H} = (\mathbb{C}^2)^{\otimes n}$, while the complementary subspace corresponds to the logical states with errors.  
The first QEC code, the Shor code~\cite{Shor:1995hbe}, encodes one logical qubit into nine physical qubits and corrects any single-qubit error.

Among various frameworks for the construction of QEC codes, the \emph{stabilizer code}~\cite{Gottesman:1997zz} is a general and widely used class.
Here, the code space $\mathcal{C}$ is defined as the simultaneous $+1$ eigenspace of a set of mutually commuting Pauli operators, called stabilizers.
Let $\mathcal{P}_n$ denote the $n$-qubit Pauli group,
\begin{align}
    \mathcal{P}_n \coloneqq \{\pm1,\pm i\} \times \langle X_1,Z_1,\ldots,X_n,Z_n\rangle,
\end{align}
where $X_i \text{ and } Z_i\, (i=1,\ldots,n)$ act as the Pauli $X$ and $Z$ operators on the $i$th qubit.  
A \emph{stabilizer group} $\mathcal{S}\subset\mathcal{P}_n$ is an Abelian subgroup that does not contain $-I$, and the code space is defined as
\begin{align}
    \mathcal{C} = \{\ket{\psi}\in\mathcal{H}\mid {}^\forall s\in\mathcal{S}, \, s\ket{\psi}=\ket{\psi}\}.
\end{align}
If the stabilizer group $\mathcal{S}$ has $m$ independent stabilizers, the dimension of the code space $\mathcal{C}$ is given by $2^{n-m}$, so the code encodes $k=n-m$ logical qubits.

In this framework, errors can be detected by measuring stabilizers: each outcome is called a \emph{syndrome}.  
For a state without any error, all syndromes are $+1$; if an error anticommutes with certain stabilizers, the corresponding syndromes turn to $-1$.  
Since the code space is the simultaneous $+1$ eigenspace of the stabilizers, any nontrivial syndrome indicates that an error has occurred.

The number of errors that can be corrected is characterized by the \emph{code distance} $d$.
Because the code space is isomorphic to the Hilbert space of $k$-qubits, there exist $k$ sets of operators $\{\overline{X}_i,\overline{Z}_i\}_{i=1}^k$ that act as Pauli $X$ and $Z$ operators on the code space.
These operators are called \emph{logical Pauli operators}.
The code distance $d$ is defined as the minimum weight (number of qubits on which it acts nontrivially) of any logical Pauli operator.
A code with the code distance $d$ can correct up to $\lfloor(d-1)/2\rfloor$ arbitrary single-qubit errors.

In general, we refer to any operator $\mathcal{O}$ that preserves the code space as a \emph{logical operator}:
\begin{align}
    {}^\forall\ket{\psi}\in\mathcal{C},\quad \mathcal{O}\ket{\psi}\in\mathcal{C}.
\end{align}
Logical operators beyond the Pauli group, such as $S$, Hadamard, CNOT, and $T$ gates, are necessary for universal quantum computation, but finding their physical implementations for a given stabilizer code is often nontrivial.

Among stabilizer codes, the \emph{Calderbank--Shor--Steane (CSS) codes}~\cite{Calderbank_1996,Steane:1996ghp} form a particularly important subclass.  
Their stabilizers are tensor products of only $X$ or only $Z$ operators: $Z$-type stabilizers detect bit-flip ($X$) errors, and $X$-type stabilizers detect phase-flip ($Z$) errors.  
Because of this separation, a CSS code can be described using two classical binary linear codes, and its logical $X$ and $Z$ operators are simple products of physical $X$ and $Z$, making analysis and implementation easier.

For a CSS code with $n_X$ $X$-type stabilizers $\{S_i^X\}_{i=1}^{n_X}$ and $n_Z$ $Z$-type stabilizers $\{S_i^Z\}_{i=1}^{n_Z}$, we can introduce binary matrices $H_X$ and $H_Z$ of size $n_X\times n$ and $n_Z\times n$, respectively:
\begin{align}
    S_i^X \eqqcolon \prod_{j=1}^{n} (X_j)^{(H_X)_{ij}}, \qquad
    S_i^Z \eqqcolon \prod_{j=1}^{n} (Z_j)^{(H_Z)_{ij}}.
\end{align}
From the commutativity of stabilizers, these matrices satisfy the following CSS condition:
\begin{align}
    H_X H_Z^{T} = 0.
\end{align}
These parity check matrices $H_X,H_Z$ completely specify the structure of the code.

CSS codes can be elegantly described using \emph{chain complexes}, establishing the so-called CSS--homology correspondence~\cite{Kitaev1997quantum,Bombin:2006cd}.  
A chain complex is a sequence of vector spaces connected by linear maps whose consecutive compositions vanish.  
We can define a chain complex $C_{css}$ that corresponds to a CSS code as
\begin{align}
    C_{css} \colon C_2 \xrightarrow{\partial_2 = H_Z^{T}} C_1 \xrightarrow{\partial_1 = H_X} C_0,
\end{align}
where $C_0=\mathbb{F}_2^{n_X}, C_1=\mathbb{F}_2^n \text{ and } C_2=\mathbb{F}_2^{n_Z}$ represent the syndromes of $X$-type stabilizers, (errors on) physical qubits, and the syndromes of $Z$-type stabilizers, respectively.
$\mathbb{F}_2$ denotes the finite field with two elements.
$\partial_1$ and $\partial_2$ are called \emph{boundary operators}, represented by the parity-check matrices $H_X$ and $H_Z^{T}$, respectively.
The condition $\partial_1\partial_2=0$ is equivalent to the CSS condition.

Elements of $C_0$, $C_1$, and $C_2$ are called 0-, 1-, and 2-chains, respectively.
In general, these algebraic objects do not admit a geometric interpretation.
However, when a chain complex arises from a cell decomposition of a manifold, it can be viewed geometrically as vertices, edges, and faces.
Such a geometric intuition is especially useful for topological codes.
For example, the chain complex of a 2D torus assigns vertices, edges, and faces to $C_0$, $C_1$, and $C_2$, respectively, giving the toric code~\cite{Kitaev:1997wr}.

Motivated by this terminology, we refer to elements of $C_0$, $C_1$, and $C_2$ as vertices, edges, and faces throughout this paper.
In this language, stabilizers take the form
\begin{align}
    \label{eq:stabilizers_in_chain_complex}
    S^Z(f) &= \prod_{e_i \in \partial f} Z_i, \qquad
    S^X(v) = \prod_{v \in \partial e_i} X_i,
\end{align}
where $f \in C_2$ is a face and $v \in C_0$ a vertex.
Thus, a $Z$-type stabilizer acts as $Z$ operators of all edges in the boundary of a face, while an $X$-type stabilizer acts as $X$ operators of all edges incident to a vertex~\footnote{
As for the correspondence with lattice gauge theory terminology, the stabilizer structure considered here is equivalent to a $\mathbb{Z}_2$ lattice gauge theory defined on a 2D hypergraph.  
The Hamiltonian takes the same form as that of the toric code~\cite{Kitaev:1997wr}:
\begin{align}
    H 
    = - \sum_{f \in C_2} S^Z(f) - \sum_{v \in C_0} S^X(v)
    = - \sum_{f \in C_2} \prod_{e \in \partial f} Z(e) - \sum_{v \in C_0} \prod_{\partial e \in v} X(e),
\end{align}
where $Z(e)$ and $X(e)$ are the Pauli $Z$ and $X$ operators acting on the edge (qubit) $e$.
The $Z$-type and $X$-type stabilizers correspond to the plaquette and star terms, respectively.  
Unlike the usual lattice on a manifold, the underlying structure here can be a general hypergraph; the incidences between edges and faces/vertices are encoded by the boundary operators $\partial_1 = H_X$ and $\partial_2 = H_Z^{T}$.  
With this identification, the code space coincides with the ground-state subspace of the above Hamiltonian.
}.

Although this geometric picture holds for manifold-based complexes, generic CSS codes are not required to admit any geometric embedding.
Their chain complexes may contain faces having only two edges or edges incident to more than two vertices, and therefore naturally live on hypergraphs rather than ordinary cell complexes.

\begin{figure}[t]
\centering
\begin{tikzpicture}[scale=1.1]
  \node[vtx,label={[vlabel]above:$v_1$}] (v1) at (0,2) {};
  \node[vtx,label={[vlabel]left:$v_2$}]  (v2) at (-1.73,-1.0) {};
  \node[vtx,label={[vlabel]right:$v_3$}] (v3) at (1.73,-1.0) {};

  \coordinate (c) at (0,0.15);
  \draw[hred,fill=none,line width=1.2pt] (v1) -- (c);
  \draw[hred,fill=none,line width=1.2pt] (v2) -- (c);
  \draw[hred,fill=none,line width=1.2pt] (v3) -- (c);
  \node[hedge] at (0.3,0.3) {$e$};

\end{tikzpicture}
\caption{
        A toy hypergraph example with one hyperedge incident to three vertices.
        Such an incidence relation is natural algebraically, but it is not the boundary of an ordinary 1-cell in a usual cell complex.
}
\label{fig:hyperedge_example}
\end{figure}

\begin{figure}[t]
\centering
\begin{tikzpicture}[scale=1.15]
  \node[vtx,red,fill=red] (v1) at (0,3.46) {};
  \node[hedge] at ($(v1)+(0.5,0)$) {$v_1,f_1$};
  \node[vtx,red,fill=red]  (v2) at (2.00,0.00) {};
  \node[hedge] at ($(v2)+(0.3,0.3)$) {$v_2,f_2$};
  \node[vtx,red,fill=red] (v3) at (-2.00,0.00) {};
  \node[hedge] at ($(v3)+(-0.3,0.3)$) {$v_3,f_3$};

  \draw[fill=none,line width=1.2pt] (v1) -- (v2);
  \node[hedge] at ($(v1)+(1.30,-1.73)$) {$e_3$};
  \draw[fill=none,line width=1.2pt] (v1) -- ($(v1)+(0.50,1.00)$);
  \node[hedge] at ($(v1)+(0.70,1.20)$) {$e_2$};
  \draw[fill=none,line width=1.2pt] (v1) -- ($(v1)+(-0.50,1.00)$);
  \node[hedge] at ($(v1)+(-0.70,1.20)$) {$e_1$};

  \draw[fill=none,line width=1.2pt] (v2) -- (v3);
  \node[hedge] at ($(v2)+(-2.00,-0.2)$) {$e_6$};
  \draw[fill=none,line width=1.2pt] (v2) -- ($(v2)+(0.87,0.00)$);
  \node[hedge] at ($(v2)+(1.20,0.00)$) {$e_4$};
  \draw[fill=none,line width=1.2pt] (v2) -- ($(v2)+(0.50,-1.00)$);
  \node[hedge] at ($(v2)+(0.70,-1.20)$) {$e_5$};

  \draw[fill=none,line width=1.2pt] (v3) -- (v1);
  \node[hedge] at ($(v3)+(0.70,1.73)$) {$e_9$};
  \draw[fill=none,line width=1.2pt] (v3) -- ($(v3)+(-0.87,0.00)$);
  \node[hedge] at ($(v3)+(-1.20,0.00)$) {$e_8$};
  \draw[fill=none,line width=1.2pt] (v3) -- ($(v3)+(-0.50,-1.00)$);
  \node[hedge] at ($(v3)+(-0.70,-1.20)$) {$e_7$};

\end{tikzpicture}
\caption{
        Hypergraph associated with one common choice of the Steane-code parity-check matrix $H_X$.
        Vertices denote $X$-stabilizers, and each qubit is a hyperedge incident on the stabilizers where the corresponding column of $H_X$ has entry $1$.}
\label{fig:steane_code_hypergraph}
\end{figure}

We first illustrate a toy example of a hyperedge incident to three vertices, shown in \Cref{fig:hyperedge_example}.
In contrast to an ordinary cell complex, where each 1-cell is incident to two vertices, such a hyperedge cannot be realized as the boundary of a single 1-cell.
Nevertheless, it is naturally described within the chain-complex formulation: it is a 1-chain whose boundary under $\partial_1$ is given by the sum of the three vertices.
This type of incidence is typical in CSS codes, where a qubit may participate in more than two $X$-type stabilizers.

As a concrete example, \Cref{fig:steane_code_hypergraph} shows the hypergraph associated with the Steane code~\cite{Steane:1996ghp}.
Here each vertex corresponds to an $X$-type stabilizer, and each hyperedge represents a qubit.
For instance, $v_1$ is incident to $e_1,e_2,e_3,e_9$, $v_2$ to $e_3,e_4,e_5,e_6$, and $v_3$ to $e_6,e_7,e_8,e_9$.
Some hyperedges (such as $e_1,e_2,e_4,e_5,e_7,e_8$) are incident to a single vertex, while others (such as $e_3,e_6,e_9$) are incident to two vertices.

Moreover, $Z$-type stabilizers can be represented by faces with analogous incidence relations, which are naturally regarded as hyperfaces rather than boundaries of ordinary 2-cells.
Taken together, these features indicate that the incidence data are more naturally described as hypergraph data than as a cellulation of a manifold.

The chain-complex viewpoint offers several advantages.  
One key benefit is that it allows properties of CSS codes to be written compactly and analyzed using algebraic-topological tools.  
For example, the logical $Z$ and $X$ operators are classified by elements of the first homology and cohomology groups:
\begin{subequations}
    \label{eq:definition_of_homology_and_cohomology_groups}
    \begin{align}
        H_1(C_\mathrm{css}) &= \ker\partial_1 / \operatorname{Im}\partial_2 = \ker H_X / \operatorname{Im} H_Z^{T},\\
        H^{1}(C_\mathrm{css}) &= \ker\partial_2^{T} / \operatorname{Im}\partial_1^{T} = \ker H_Z / \operatorname{Im} H_X^{T}.
    \end{align}
\end{subequations}
The code distance $d$ is determined by the nontrivial smallest-weight elements of these groups, known as the algebraic topology notion of (co)systoles.
We reproduce this fact in the gauge field formalism in the next subsection.

Furthermore, chain-complex methods provide a systematic route to design new families of CSS codes.  
Tensor-product complexes produce hypergraph product codes~\cite{Tillich_2014}, while lifted-product constructions~\cite{Panteleev:2019upy,Panteleev:2021wvc} and fiber-bundle constructions~\cite{Hastings:2020jbo} can be utilized to generate quantum LDPC codes with favorable distance and scaling.
Since these codes are defined by algebraic ways, their logical operators and distance can be analyzed using homological algebra.

Finally, this correspondence between CSS codes and chain complexes provides the foundation for applying the \emph{gauge field formalism} to analyze logical gates, which will be introduced in the next subsection.

\subsection{Gauge field formalism}

In this subsection, we explain the gauge field formalism for the quantum CSS codes.  
This formalism has recently attracted attention to provide a unified way to describe physical and logical gates as operator-valued cochains defined on the underlying chain complex.

We first introduce the \emph{electric gauge field} $a$, defined as a linear map from $C_1$ to $u(2^n)$ (the Lie algebra of the $2^n$D unitary group $U(2^n)$)
\footnote{
    Briefly speaking, $u(2^n)$ denotes the real vector space of all $2^n \times 2^n$ Hermitian matrices, equipped with the Lie bracket as its algebraic operation.  
    For $u(2^n)$, this Lie bracket is simply the commutator, $[A,B] = (AB - BA)$, which endows $u(2^n)$ with the structure of a Lie algebra.
}.
Let us denote the basis of $C_1$ as $\{e_j\}_{j=1}^n$, where each $e_j$ corresponds to the $j$th qubit.
The value of $a$ on the basis element $e_j$ is defined as
\begin{align}
    a(e_j) = 
    \overbrace{I_2 \otimes \cdots \otimes I_2}^{j-1}
    \otimes z \otimes
    \overbrace{I_2 \otimes \cdots \otimes I_2}^{n-j}
    \eqqcolon z_j,
\end{align}
where $z \coloneqq \mqty(0 & 0\\0 & 1)$.  
Similarly, we define the \emph{magnetic gauge field} $b$ as
\begin{align}
    b(e_j) =
    \overbrace{I_2 \otimes \cdots \otimes I_2}^{j-1}
    \otimes x \otimes
    \overbrace{I_2 \otimes \cdots \otimes I_2}^{n-j}
    \eqqcolon x_j,
\end{align}
where $x \coloneqq \mqty(1/2 & -1/2 \\ -1/2 & 1/2)$.  
The relations between the Pauli matrices and these fields are given by
\begin{align}
    \label{eq:relation_between_Pauli_and_gauge_field}
    Z = \exp(i\pi z) = I - 2z, 
    \qquad
    X = \exp(i\pi x) = I - 2x.
\end{align}
Because these gauge fields $a$ and $b$ can be regarded as linear maps from 1-chains (qubits) to operators (matrices), they are referred to as \emph{operator-valued cochains}, which are called \emph{gauge fields} in quantum field theory.
This is the origin of the term \emph{gauge field formalism}.

Using these gauge fields, we can define the unitary operators $\overline{Z}(\gamma)$ and $\overline{X}(\gamma)$ as
\footnote{
In quantum field theory, this operator corresponds to a Wilson loop in high-energy physics or a Berry phase in condensed-matter physics, typically written as
\begin{align}
    \overline{Z}(\gamma) \leftrightarrow \exp(i \oint_\gamma dx\, a(x)),
\end{align}
where $a(x)$ is the $U(1)$ gauge field.
}
\begin{subequations}
    \label{eq:definition_of_trasversal_Pauli_operators}
    \begin{align}
        \label{eq:definition_of_barZ}
        \overline{Z}(\gamma) 
        &= \exp(i\pi a(\gamma))
         = \exp(i\pi \sum_i \gamma^i a(e_i))
         = \prod_{i=1}^n (Z_i)^{\gamma^i},\\
        \label{eq:definition_of_barX}
        \overline{X}(\gamma)
        &= \exp(i\pi b(\gamma))
         = \exp(i\pi \sum_i \gamma^i b(e_i))
         = \prod_{i=1}^n (X_i)^{\gamma^i},
    \end{align}
\end{subequations}
where $\gamma = \sum_{i=1}^n \gamma^i e_i$ is a 1-chain with $\gamma^i \in \{0,1\}$.
It is easily seen that these operators satisfy $\overline{Z}(\gamma_1) \overline{Z}(\gamma_2) = \overline{Z}(\gamma_1 + \gamma_2)$.
The commutation relations between the gauge fields $a,b$ with a 1-chain $\gamma_1$ and these Pauli operators $\overline{X}, \overline{Z}$ with a 1-chain $\gamma_2$ are given by
\begin{subequations}
    \begin{align}
        \overline{X}(\gamma_2) a(\gamma_1) \overline{X}(\gamma_2) &= a(\gamma_1) - 2 a(\gamma_1 \cap \gamma_2) + \inp{\gamma_1}{\gamma_2},\\
        \overline{Z}(\gamma_2) b(\gamma_1) \overline{Z}(\gamma_2) &= b(\gamma_1) - 2 b(\gamma_1 \cap \gamma_2) + \inp{\gamma_1}{\gamma_2},
    \end{align}
\end{subequations}
where $\inp{\gamma_1}{\gamma_2} = \sum_i \gamma_1^i \gamma_2^i$ is the inner product of $\gamma_1$ and $\gamma_2$, and $\gamma_1 \cap \gamma_2 = \sum_i \gamma_1^i \gamma_2^i e_i$ denotes their intersection.
The inner product $\inp{\gamma_1}{\gamma_2}$ takes integer values and should be interpreted as $\inp{\gamma_1}{\gamma_2}I_{2^n}$ when added to $a$ or $b$, which takes values in the Lie algebra $u(2^n)$.

Now, let us connect this formalism with quantum error correction, in particular the quantum CSS code.  
Using the gauge fields, the stabilizers \Cref{eq:stabilizers_in_chain_complex} can be expressed as
\begin{align}
    S^Z(f) = \overline{Z}(\del_2 f), \qquad S^X(v) = \overline{X}(\del_1^T v),
\end{align}
where $\del_1^T$ is the transpose of $\del_1$, with $\sum_{v \in \del e_i} e_i = \del_1^T v$.

    In this formalism, errors can be interpreted as gauge field excitations.
    For a $Z$ error chain $\gamma$, anticommutation with an $X$-type stabilizer $S^X(v)$ flips its eigenvalue to $-1$, contributing to an energy penalty since the Hamiltonian is a sum of stabilizers.
    Writing the error operator as $\overline{Z}(\gamma)$ and denoting by $n_{\mathrm{anti}}$ the number of anticommuting stabilizers, one has
    \begin{align}
        (H - E_0)\,\overline{Z}(\gamma)\ket{\psi}
        = 2 n_{\mathrm{anti}}\,\overline{Z}(\gamma)\ket{\psi},
    \end{align}
    where $E_0 \coloneqq -n_X - n_Z$.
    The commutation relation with $S^X(v)$ is given by
    \footnote{
    We can introduce a linear map $\tilde{v}: C_0 \to \mathbb{Z}$ defined by $\tilde{v}(v') = \inp{v}{v'}$ for each vertex $v$, which is interpreted as a 0-cochain.
    Defining the exterior derivative $d\tilde{v}(\gamma)$ as $d \tilde{v}(\gamma) \coloneqq \inp{v}{\del_1 \gamma}$, we can rewrite this relation as
    \begin{align}
        S^X(v)\,\overline{Z}(\gamma)\,S^X(v)
        = \exp(i\pi (a(\gamma) + d \tilde{v}(\gamma))).
    \end{align}
    This relation shows that conjugation with $S^X(v)$ induces a gauge transformation of the electric field:
    \begin{align}
        \label{eq:gauge_transformation_of_a}
        a \to a + d\tilde{v}.
    \end{align}
    The commutation relation then implies that $\overline{Z}(\gamma)$ becomes a logical operator if and only if it is gauge invariant.

    This calculation holds for the exponential of an integer function $f(a)$ of the gauge field, i.e. $S^{X}(v) \exp(i\pi f(a))S^{X}(v) = \exp(i\pi f(a + d\tilde{v}))$, as pointed out in Ref.~\cite{hsin2024classifyinglogicalgatesquantum}.
    However, it is no longer valid for more general functions, and the general transformation is given by $S^{X}(v) g(a(\gamma)) S^{X}(v) = g(a(\gamma)-2a(\gamma \cap \del_1^T v) + d\tilde{v}(\gamma))$ from the commutation relation between $a$ and $S^X(v)$.
    }
    \begin{align}
        \label{eq:commutation_relation_between_barZ_and_SX}
        S^X(v)\,\overline{Z}(\gamma)\,S^X(v)
        = \exp\!\bigl(i\pi \inp{\gamma}{\del_1^T v}\bigr)\,\overline{Z}(\gamma)
        = \exp\!\bigl(i\pi \inp{\del_1 \gamma}{v}\bigr)\,\overline{Z}(\gamma).
    \end{align}
    Hence, the syndrome of $S^X(v)$ is $-1$ if and only if $\inp{\del_1 \gamma}{v} = 1$; i.e. $v$ lies on the boundary of $\gamma$.
    Therefore, a $Z$ error on a 1-chain $\gamma$ creates syndrome defects on $\del_1 \gamma$, which can be interpreted as gauge-charge excitations localized at the boundary.
    Dually, $X$ errors create $Z$-type defects.
    In geometrically realized models such as the toric code, these correspond to electric and magnetic anyonic excitations.

The logical $Z$ and $X$ operators can be expressed compactly in terms of the gauge fields, too.
As we mentioned in \Cref{eq:definition_of_homology_and_cohomology_groups}, these operators are classified by the first homology and cohomology groups.
Let us derive this fact using the gauge field formalism while seeing how this formalism works.

To identify when $\overline{Z}(\gamma)$ or $\overline{X}(\gamma)$ acts as a logical operator, we examine their commutation relations with the stabilizers.
Focusing on $\overline{Z}(\gamma)$, the commutation with $S^Z(f)$ is trivially satisfied.
As for the $X$-type stablizers, commutativity \Cref{eq:commutation_relation_between_barZ_and_SX} with them requires that 
\begin{align}
    \label{eq:condition_for_logical_Z_operator}
    {}^\forall v \in C_0, \quad
    0 = \inp{\gamma}{\del_1^T v} = \inp{\del_1 \gamma}{v} \quad (\mathrm{mod}\,2)
    \quad \Leftrightarrow \quad
    \del_1 \gamma = 0
    \quad \Leftrightarrow \quad
    \gamma \in \ker \del_1.
\end{align}
Using a similar argument, the logical $X$ operators are characterized by $\overline{X}(\tilde{\gamma})$ with $\tilde{\gamma} \in \ker \del_2^T$.

Next, we have to consider how many nontrivial logical $Z$ operators exist on the code space, characterizing the type of the logical qubits.
This is because $\overline{Z}(\gamma)$ may act trivially on the code space even if $\gamma \in \ker \del_1$.
From \Cref{thm:IdentityCondition} in \Cref{sec:IdentityCondition}, we know that $\overline{Z}(\gamma)$ acts as the identity operator on the code space if it commutes with all logical $X$ operators.
Their commutation with the logical $X$ operators with $\tilde{\gamma} \in \ker \del_2^T$ is given by
\begin{align}
    \label{eq:commutation_relation_between_barZ_and_barX}
    \overline{X}(\tilde{\gamma})\,\overline{Z}(\gamma)\,\overline{X}(\tilde{\gamma})
    = \exp(i\pi \inp{\gamma}{\tilde{\gamma}})\,\overline{Z}(\gamma).
\end{align}
Thus, $\overline{Z}(\gamma)$ and $\overline{X}(\tilde{\gamma})$ commute if and only if
$\inp{\gamma}{\tilde{\gamma}} = 0$.
 Decompose $\gamma$ and $\tilde{\gamma}$ with some face $f \in C_2$ and vertex $v \in C_0$ as
\begin{align}
    \gamma = h(\gamma) + \del_2 f, 
    \qquad
    \tilde{\gamma} = ch(\tilde{\gamma}) + \del_1^T v,
\end{align}
where $h(\gamma) \in H_1(C_\mathrm{css})$ and $ch(\tilde{\gamma}) \in H^1(C_\mathrm{css})$ denote the representatives of the homology and cohomology classes of $\gamma$ and $\tilde{\gamma}$, respectively.
Then, we have 
\begin{align}
    \inp{\gamma}{\tilde{\gamma}} = \inp{h(\gamma)}{ch(\tilde{\gamma})}
\end{align}
since other terms vanish due to $\del_1 h(\gamma) = \del_2^T ch(\tilde{\gamma}) = \del_1\del_2=0$.
By the Poincar\'{e} duality, there exists $\gamma' \in C_1$ with $\inp{h(\gamma)}{ch(\gamma')}=1$ if $h(\gamma)$ is nontrivial.
Hence, the only way for the logical $Z$ operator $\overline{Z}(\gamma)$ to commute with all logical $X$ operators $\overline{X}(\tilde{\gamma})$ is $h(\gamma)=0$.
Therefore, we find that any logical $Z$ operators with $\gamma = \del_2 f \in \Im \del_2$ act trivially on the code space.
From the point of view of the quantum error correction, this condition means that $\overline{Z}(\del_2 f)$ can be expressed as a product of the $Z$-type stabilizers.

Having obtained these results, we can classify the logical $Z$ operators.
From \Cref{eq:condition_for_logical_Z_operator}, we know that $\overline{Z}(\gamma)$ is a logical operator if $\gamma \in \ker \del_1$.
Furthermore, from the above discussion, two logical $Z$ operators $\overline{Z}(\gamma_1)$ and $\overline{Z}(\gamma_2)$ represent the same logical operation on the code space if and only if
$\gamma_1 - \gamma_2 \in \Im\,\del_2$, since $\overline{Z}(\gamma_1) = \overline{Z}(\gamma_1-\gamma_2)\,\overline{Z}(\gamma_2)$.
Consequently, logical $Z$ operators are classified by $\ker \del_1 / \Im \,\del_2$, which is in agreement with the first homology group $H_1(C_\mathrm{css})$.
Analogously, the logical $X$ operators are classified by the first cohomology group $H^1(C_\mathrm{css})$.

So far, we have seen how the gauge field formalism describes the logical Pauli operators and relates them to the homology and cohomology groups.
Recently, this formalism has been extended to analyze logical gates beyond Pauli operators~\cite{Zhu:2023xfg,Lin:2024uhb,Golowich:2024ogv,breuckmann2024cupsgatesicohomology,hsin2024classifyinglogicalgatesquantum,zhu2025topologicaltheoryqldpcnonclifford,Kobayashi:2025cfh}.
However, despite these advances, a systematic and practical method for decomposing logical gates into physical gates using this formalism has not been fully established.
In addition, their logical gates are mainly limited to certain specific quantum codes or specific families of chain complexes with particular structures, such as cup products.
In the next section, we further develop this formalism to systematically construct and decompose logical gates with physical gates, which is applicable to general CSS codes.

\section{Logical Gate and Gauge Field Formalism}
\label{sec:logical_gate_and_gauge_field_formalism}

In this section, we study logical gates within the gauge field formalism.
Our goal is to clarify how logical operations can be systematically constructed and decomposed using gauge fields.
Here, we \emph{do not} assume any special structure of the underlying chain complex, such as cup products or manifold structure; rather, we consider general CSS codes defined on chain complexes on hypergraphs, which may not have a geometric embedding.
In exchange for this applicability, our construction does not care about the fault tolerance of logical gates, or the efficiency of their decompositions, such as constant-depth or optimal number of physical gates.
Our aim is to provide a general framework for constructing and decomposing logical gates in terms of physical gates using the gauge field formalism.
We leave incorporating fault tolerance and efficiency into this framework as an important direction for future research.

Our construction is based on three key ideas.
The first idea is that any function $g(\overline{X}, \overline{Z})$ of the logical $X$ and $Z$ operators represents a valid logical operator.  
This follows from the fact that the logical $X$ and $Z$ operators commute with all stabilizers, and thus so as any function of them:
\begin{align}
    \comm{g(\overline{X}, \overline{Z})}{S^Z(f)} = \comm{g(\overline{X}, \overline{Z})}{S^X(v)} = 0,
\end{align}
for any $f \in C_2$ and $v \in C_0$.

The second idea is that if a physical gate $\mathcal{O}$ can be expressed in terms of the physical $X$ and $Z$ operators as $\mathcal{O} = g(X, Z)$, its logical counterpart $\overline{\mathcal{O}}$ can be obtained simply by replacing the physical operators with the logical ones:
\begin{align}
\label{eq:logical_op_by_replacement}
    \overline{\mathcal{O}} = g(\overline{X}, \overline{Z}).   
\end{align}
This prescription works because the commutation relations between $X_i$ and $Z_j$, given by
$X_i Z_j = (-1)^{\delta_{ij}} Z_j X_i$ for any $i, j$, also hold for the logical operators $\overline{X}_i$ and $\overline{Z}_i$.  
Hence, any algebraic relation that defines a physical gate is automatically preserved for the corresponding logical gate.

The third idea is that if an operator $\mathcal{O}$ has a decomposition with the Pauli $X$ and $Z $ operators, such as the Euler-angle decomposition, its logical version can be decomposed into physical gates using the relations
\begin{align}
    (X)^c = I - 2c x, \qquad (Z)^c = I - 2c z,
\end{align}
where $c \in \{0,1\}$.  
Although this approach is somewhat abstract and heuristic, we will make its meaning concrete in the following subsections.  
We will see that the gauge field formalism provides a simple and transparent way to compute commutation relations among operators, thereby clarifying the structure of logical gates.

As a first and illustrative example, we examine how to construct and decompose the logical $S$ gate within this formalism.  
Subsequently, we discuss the rest of the Clifford gates, the logical $H$ and controlled-$Z$ gates.
Then, as an application of the formalism beyond the Clifford gates, we explore the decomposition of the logical multicontrolled-$Z$ and $T$ gates.  
Finally, we give some comments on our methods and results.

\subsection{S gate}

Let us consider the $S$ gate, defined as 
\begin{align}
    \label{eq:physical_S_gate}
    S \coloneqq \mqty(1 & 0 \\ 0 & i) = \exp(i\frac{\pi}{2}\frac{I-Z}{2}).
\end{align}
According to the replacement rule in \Cref{eq:logical_op_by_replacement}, the logical $S$ gate for a single logical qubit $\gamma \in H_1(C)$ is given by
\begin{align}
    \label{eq:logical_S_gate}
    \overline{S}(\gamma) = \exp(i\frac{\pi}{2}\frac{I-\overline{Z}(\gamma)}{2}).
\end{align}

We first verify that this operator indeed acts as the $S$ gate on the code space.
The logical $S$ gate should satisfy the same commutation relations with the logical Pauli operators as the physical $S$ gate.
For the physical $S$ gate in \Cref{eq:physical_S_gate}, the commutation relation with the Pauli $X$ operator is given by
\begin{align}
    X S X = i Z S.
    \label{eq:commutation_relation_of_S_gate_with_X}
\end{align}
Let us check that the logical $S$ gate in \Cref{eq:logical_S_gate} satisfies the same relation with the logical $X$ operator.
Let $\gamma \in H_1(C)$ and $\tilde{\gamma} \in H^1(C)$ be a dual representative satisfying $\inp{\gamma}{\tilde{\gamma}} = 1$.  
The commutation relation between $\overline{S}(\gamma)$ and $\overline{X}(\tilde{\gamma})$ is computed as
\begin{subequations}
    \label{eq:commutation_relation_of_logical_S_gate_with_logical_X}
    \begin{align}
        \overline{X}(\tilde{\gamma})\, \overline{S}(\gamma)\, \overline{X}(\tilde{\gamma})
        &= \exp\!\left(i\frac{\pi}{2}\frac{I-\overline{X}(\tilde{\gamma})\overline{Z}(\gamma)\overline{X}(\tilde{\gamma})}{2}\right)\\
        &= \exp\!\left(i\frac{\pi}{2}\frac{I - (-\overline{Z}(\gamma))}{2}\right)\\
        &= i\,\overline{Z}(\gamma)\,\overline{S}(\gamma).
    \end{align}
\end{subequations}
Here, we used \Cref{eq:commutation_relation_between_barZ_and_barX} and the fact that $\exp(i\pi \overline{Z}(\gamma)/2) = i\overline{Z}(\gamma)$ since $\overline{Z}(\gamma)^2 = I$.  
Thus, $\overline{S}(\gamma)$ behaves as the $S$ gate on the code space.
If we take a different logical qubit $\tilde{\gamma}'$ satisfying $\inp{\gamma}{\tilde{\gamma}'} = 0$, the same calculation shows that $\overline{S}(\gamma)$ commutes with $\overline{X}(\tilde{\gamma}')$, as expected.

Next, we confirm that $\overline{S}(\gamma)$ commutes with all stabilizers.  
Although this follows directly from its construction, it is instructive to check explicitly.  
Since commutativity with $S^Z(f)$ is obvious, we consider that with $S^X(v)$ for an arbitrary vertex $v$:
\begin{align}
    S^X(v)\, \overline{S}(\gamma)\, S^X(v)
    &= \exp\!\left(i\frac{\pi}{2}\frac{I - \overline{X}(\del_1^T v)\,\overline{Z}(\gamma)\,\overline{X}(\del_1^T v)}{2}\right)\\
    &= \exp\!\left(i\frac{\pi}{2}\frac{I - \exp(i\pi \inp{\gamma}{\del_1^T v})\,\overline{Z}(\gamma)}{2}\right)\\
    &= \overline{S}(\gamma).
\end{align}
Here, we used \Cref{eq:commutation_relation_between_barZ_and_barX} and the fact that $\exp(i\pi \inp{\gamma}{\del_1^T v}) = \exp(i\pi \inp{\del_1 \gamma}{v}) = 1$ for $\gamma \in H_1(C)$.  
Thus, $\overline{S}(\gamma)$ commutes with all stabilizers and therefore qualifies as a logical operator.  
The structure of this calculation parallels that for logical operators—the only difference is replacing $\tilde{\gamma}$ by $\del_1^T v$.

Now, let us analyze how this logical $S$ gate can be decomposed into physical gates.  
For $\gamma = \sum_i \gamma^i e_i$, we can expand $\overline{Z}(\gamma)$ using \Cref{eq:relation_between_Pauli_and_gauge_field}:
\begin{align}
    \overline{Z}(\gamma)
    = \prod_i Z_i^{\gamma^i}
    = \prod_i (I - 2\gamma^i z_i)
    = I - 2\sum_i \gamma^i z_i + 4\sum_{i<j} \gamma^i \gamma^j z_i z_j -8 \Delta,
\end{align}
where $\Delta$ is defined as
\begin{align}
    \Delta \coloneqq
    \sum_{i_1<i_2<i_3}\prod_{k=1}^3 \gamma^{i_k} z_{i_k}
    - 2\sum_{i_1<i_2<i_3<i_4}\prod_{k=1}^4 \gamma^{i_k} z_{i_k}
    + \cdots
    + (-2)^{n-3}\!\!\sum_{i_1<\cdots<i_n}\prod_{k=1}^n \gamma^{i_k} z_{i_k}.
\end{align}
Hence, the exponent in \Cref{eq:logical_S_gate} can be written as
\begin{align}
    \frac{I-\overline{Z}(\gamma)}{2}
    = \sum_i \gamma^i z_i - 2\sum_{i<j} \gamma^i \gamma^j z_i z_j + 4\Delta.
\end{align}
The logical $S$ gate can thus be expressed as
\begin{align}
    \label{eq:logical_S_gate_decomposition}
    \overline{S}(\gamma)
    &= \exp(i\frac{\pi}{2}\left(\sum_i \gamma^i z_i - 2\sum_{i<j}\gamma^i \gamma^j z_i z_j + 4\Delta\right))\\
    &= \prod_i \exp(i\frac{\pi}{2}\gamma^i z_i)\,
       \prod_{i<j}\exp(-i\pi\gamma^i \gamma^j z_i z_j)\,
       \exp(2\pi i\Delta).
\end{align}
All terms in the exponential commute because each $z_i$ is diagonal in the computational basis.  
The first product corresponds to the physical $S$ gates, while the second corresponds to the physical $CZ$ gates (see \Cref{eq:physical_cz_gate}).  
The final term vanishes because $\Delta$ is a diagonal matrix with integer eigenvalues, implying $\exp(2i\pi\Delta)=I$.
Therefore, we find that the logical $S$ gate can be decomposed into the physical $S$ and $CZ$ gates as
\begin{align}
    \overline{S}(\gamma)
    = \prod_i (S_i)^{\gamma^i}\,
      \prod_{i<j} (CZ_{ij})^{\gamma^i \gamma^j}.
\end{align}
Let us rewrite this expression more explicitly with respect to the physical qubits.
Define the support of $\gamma$ as $\supp(\gamma) \coloneqq \{ j \in \{1,\ldots, n\} \mid \gamma^j = 1 \}$, and here we denote it as $\supp(\gamma) \eqqcolon \{j_1,\ldots,j_m\}$ with $j_k<j_{k+1}$.
Then, the above expression can be rewritten as
\begin{align}
    \label{eq:logical_S_gate_by_physical_gates}
    \overline{S}(\gamma)
    = \prod_{k=1}^m S_{j_k}
      \prod_{1\le k_1<k_2\le m} CZ_{j_{k_1}j_{k_2}}.
\end{align}
Hence, we find that the logical $S$ gate is realized by the transversal $S$ gates on the support of $\gamma$ and the $CZ$ gates between every pair within that support.
This construction is consistent with recent work on lift-connected surface codes~\cite{Old:2025lsb}.

We can express this result compactly using the gauge fields.  
Using Eqs.~\eqref{eq:useful_formulae_a_gamma} in \Cref{app:useful_formulae}, we can rewrite \Cref{eq:logical_S_gate_decomposition} as~\footnote{
    Based on the correspondence with the $U(1)$ gauge theory, the counterpart to this logical $S$ gate may be written as
    \begin{align}
        \overline{S}(\gamma) \leftrightarrow \exp(\frac{i}{2} \biggl(\oint_\gamma dx\, a(x)\biggr)^2),
    \end{align}
    which looks somewhat strange since it involves the square of the loop integral of the gauge field.
    Note that this is \emph{not} equivalent to the exponential of the square of the gauge field itself, i.e. $\exp(i/2 \oint_\gamma dx\, a(x)^2)$.

}
\begin{align}
    \overline{S}(\gamma)
    = \exp(i\frac{\pi}{2}a(\gamma))
      \exp(i\frac{\pi}{2}(a(\gamma)^2 - a(\gamma)))
    = \exp(i\frac{\pi}{2} a(\gamma)^2).
\end{align}
This expression is useful when we calculate the commutation relations with other operators or show its (co)homology dependence on the code space.

We now comment on the gauge-theoretic and geometric interpretation of the logical $S$ gate. 
In contrast to the physical $S$ gate in \Cref{eq:physical_S_gate}, the logical operator is expressed as the exponential of the square of the gauge field, $\exp(i\frac{\pi}{2} a(\gamma)^2)$. 
As shown in the derivation, this expression should be understood as a compact representation of a sum of single-qubit and two-qubit diagonal phase terms, which decompose into physical $S$ gates and $CZ$ gates, respectively.

From the viewpoint of gauge theory, a Wilson line operator is typically given by the exponential of a line integral of the gauge field, without involving higher powers. 
In the present context, however, the quadratic term $a(\gamma)^2$ naturally encodes two-body interactions among qubits in the support of $\gamma$. 
These interaction terms are essential to ensure that $\overline{S}(\gamma)$ is gauge invariant and hence acts as a valid logical operator. 
In this sense, the square of the gauge field should not be interpreted as the square of an independent Wilson line, but rather as a compact way to organize the multiqubit diagonal phase interactions required for logical consistency.

The relation to geometry becomes particularly transparent in the toric code, which can be regarded as a geometrically realized special case of our formalism. 
In this case, $\gamma$ corresponds to a nontrivial cycle on the torus, i.e. a noncontractible loop. 
Then, the logical operator $\overline{S}(\gamma)$ is decomposed into transversal $S$ gates acting on all qubits along the loop and all-to-all $CZ$ gates between every pair of qubits on the same support. 
Thus, the logical $S$ gate consists of both single-qubit phases along the loop and pairwise interactions that enforce the global topological constraint associated with the nontrivial cycle.

This geometric interpretation relies on the existence of a spatial embedding of the chain complex. 
In contrast, for general hypergraph CSS codes, we do not assume any geometric locality or embedding structure, and the construction is purely algebraic in terms of the chain complex. 
Nevertheless, the same gauge field expression remains valid and provides a unified description of logical gates beyond geometrically realized codes.

Finally, let us examine how the operation of $\overline{S}(\gamma)$ acts on the code space.
As we reviewed in \Cref{sec:Review}, the logical $Z$ operators with $\gamma_1$ and $\gamma_2$ of the same homology class are logically equivalent.
This property is critical for the physical gates to act as a logical gate, since these 1-chains represent different physical qubits but the \emph{same} logical qubit.
Using the gauge field expression, we can explicitly see that this property holds for the logical $S$ gate, which shows that the logical operation of this gate depends only on the homology class of its argument.
We illustrate this in the following.

Let $\gamma_i \in \ker \del_1$ for $i=1,2$ and $f \in C_2$ with $\gamma_1 = \gamma_2 + \del_2 f$, which means that $\gamma_1$ and $\gamma_2$ belong to the same homology class.
Substituting this into the gauge field expression yields
\begin{align}
    \overline{S}(\gamma_1)
    = \overline{S}(\gamma_2)\, \overline{S}(\del_2 f)\,
      \exp(i\pi\, a(\gamma_2)\, a(\del_2 f)).
\end{align}
The cross term $\exp(i\pi a(\gamma_2)a(\del_2 f))$ behaves trivially on the code space from \Cref{thm:ControlledLogicallyIdentityOperator} in \Cref{app:cohomology_dependence_of_logical_gates}.
Intuitively, this is because it is a product of a controlled version of the stabilizer $S^Z(f)$.
We can show that $\overline{S}(\del_2 f)$ is also logically trivial.
This is because this term can be evaluated by \Cref{eq:logical_S_gate} as
\begin{align}
    \overline{S}(\del_2 f)
    = \exp(i\frac{\pi}{2}\frac{I - S^Z(f)}{2}),
\end{align}
and $S^Z(f)$ acts as the identity on the code space.
Alternatively, we can check that $\overline{S}(\del_2 f)$ commutes with all stabilizers and logical operators, which is given in \Cref{app:cohomology_dependence_of_logical_gates}.
Thus, we conclude that the logical $S$ gates with $\gamma_1$ and $\gamma_2$ are logically equivalent:
\begin{align}
    \overline{S}(\gamma_1) \sim \overline{S}(\gamma_2).
\end{align}
Here, we define the logical equivalence $\sim$ as equality on the code space, as is defined in \Cref{app:cohomology_dependence_of_logical_gates} in detail.
This means that the logical operation of $\overline{S}(\gamma)$ depends only on the homology class of $\gamma$:
\begin{align}
    \overline{S}(\gamma) \sim \overline{S}(h(\gamma)),
\end{align}
where $h(\gamma) \in H_1(C)$ is the homology class representative of $\gamma$.
This property ensures that our construction is a valid logical $S$ gate.

\subsection{Hadamard gate}

Let us now examine the Hadamard gate $H$ in a manner analogous to the $S$ gate.
This gate is defined as
\begin{align}
    H \coloneqq \frac{1}{\sqrt{2}}\mqty(1 & 1 \\ 1 & -1).
\end{align}
The Euler-angle decomposition of this gate is given by
\begin{align}
    H &= \exp(i\frac{\pi}{2}) \exp(-i\frac{\pi}{4} Z) \exp(-i\frac{\pi}{4} X) \exp(-i\frac{\pi}{4} Z) \\
      &= \exp(-i\frac{\pi}{4}) 
      \exp(i\frac{\pi}{2}\frac{I-Z}{2}) 
      \exp(i\frac{\pi}{2}\frac{I-X}{2}) 
      \exp(i\frac{\pi}{2}\frac{I-Z}{2}).
\end{align}
Following the same prescription as before, the logical Hadamard gate acting on a single logical qubit $\gamma \in H_1(C_\mathrm{css}    )$ can be expressed as
\begin{align}
    \overline{H}(\gamma) 
    &= \exp(-i\frac{\pi}{4}) 
       \exp(i\frac{\pi}{2}\frac{I-\overline{Z}(\gamma)}{2})
       \exp(i\frac{\pi}{2}\frac{I-\overline{X}(\tilde{\gamma})}{2})
       \exp(i\frac{\pi}{2}\frac{I-\overline{Z}(\gamma)}{2}),
\end{align}
where $\tilde{\gamma}$ denotes a 1-chain whose cohomology class is dual to $\gamma$, satisfying $\inp{\gamma}{\tilde{\gamma}} = 1$.
We can verify that this operator behaves as a Hadamard gate on the code space by checking its commutation relations with the logical $X$ and $Z$ operators, similar to the case of the $S$ gate.

To obtain a decomposition into physical gates, let us expand each exponential term.  
The second and last terms are already known from the logical $S$ gate construction.  
For the third term, we evaluate it in the same manner as $\overline{Z}(\gamma)$:
\begin{align}
    \exp(i\frac{\pi}{2}\frac{I-\overline{X}(\tilde{\gamma})}{2})
    &= \exp(i\frac{\pi}{2}
    \biggl(\sum_{i=1}^n \tilde{\gamma}^i x_i 
    + 2\sum_{i<j} \tilde{\gamma}^i \tilde{\gamma}^j x_i x_j
    \biggr)
    ).
    \label{eq:decomposition_of_logical_X_gate_part}
\end{align}
Using the single-qubit identity $x = H z H$, this expression can be rewritten as
\begin{align}
    \exp(i\frac{\pi}{2}\frac{I-\overline{X}(\tilde{\gamma})}{2})
    &= H(\tilde{\gamma}) 
       \exp(i\frac{\pi}{2}
       \biggl(\sum_{i=1}^n \tilde{\gamma}^i z_i 
       + 2\sum_{i<j} \tilde{\gamma}^i \tilde{\gamma}^j z_i z_j\biggr)
       )
       H(\tilde{\gamma}) \\
    &= H(\tilde{\gamma})\,\overline{S}(\tilde{\gamma})\,H(\tilde{\gamma}),
\end{align}
where we have defined the transversal Hadamard gates on the support of $\tilde{\gamma}$ as
\begin{align}
    H(\tilde{\gamma}) \coloneqq \prod_{i=1}^n (H_i)^{\tilde{\gamma}^i}.
\end{align}
Substituting these results, the logical Hadamard gate can be decomposed as
\begin{align}
    \overline{H}(\gamma)
    = e^{-i\frac{\pi}{4}}\, 
      \overline{S}(\gamma)\, 
      H(\tilde{\gamma})\, 
      \overline{S}(\tilde{\gamma})\, 
      H(\tilde{\gamma})\, 
      \overline{S}(\gamma).
\end{align}
This decomposition shows that the logical Hadamard gate can be implemented by alternating layers of the logical $S$ gates with $\gamma$ or $\tilde{\gamma}$ and the transversal physical Hadamard gates with $\tilde{\gamma}$.

Next, let us rewrite this gate in terms of the gauge fields. Similarly to the case of the logical $S$, we can express \Cref{eq:decomposition_of_logical_X_gate_part} as
\begin{align}
    \exp(i\frac{\pi}{2}\frac{I-\overline{X}(\tilde{\gamma})}{2})
    = \exp(i\frac{\pi}{2} b(\tilde{\gamma})^2).
\end{align}
Consequently, the logical Hadamard gate is written as 
\begin{align}
    \label{eq:gauge_field_expression_of_logical_H_gate}
    \overline{H}(\gamma)
    = \exp(-i\frac{\pi}{4})
      \exp(i\frac{\pi}{2} a(\gamma)^2)
      \exp(i\frac{\pi}{2} b(\tilde{\gamma})^2)
      \exp(i\frac{\pi}{2} a(\gamma)^2).
\end{align}

The logical operation of this logical Hadamard gate depends only on the homology class of $\gamma$, similar to the logical $S$ gate.
We show the detailed proof not here but in \Cref{app:cohomology_dependence_of_logical_gates}.
In short, it follows the same reasoning as for the logical $S$ gate.
The second and fourth terms in \Cref{eq:gauge_field_expression_of_logical_H_gate} are the same as the logical $S$ gate, which we have already shown to depend only on the homology class.
The third term can be shown to depend only on the cohomology class by similar reasoning, since it is of a similar form to the logical $S$ gate replaced by the magnetic gauge field.
In this way, we can show that
\begin{align}
    \overline{H}(\gamma_1) 
    \sim \overline{H}(\gamma_2),
\end{align}
with $\gamma_1, \gamma_2 \in \ker \del_1$ satisfying $\gamma_1 = \gamma_2 + \del_2 f$ for some $f \in C_2$.
This confirms that the logical Hadamard gate is well defined in the code space.

\subsection{(Multi)controlled-Z gate}

In this subsection, we consider the controlled-$Z$ (CZ) gate and its generalization, the multicontrolled-$Z$ gate.
We begin by recalling the definition of the physical CZ gate between qubits $i$ and $j$:
\begin{align}
    \label{eq:physical_cz_gate}
    CZ_{i,j} 
    = I_j P^+_i + Z_j P^-_i
    = \exp\left(i\pi \frac{I-Z_i}{2} \frac{I-Z_j}{2}\right),
\end{align}
where $P^{\pm}_i \coloneqq (I\pm Z_i)/2$ are the projectors onto the computational basis states $\ket{0}$ and $\ket{1}$ of the $i$th qubit.
Using this form, the logical CZ gate acting on two logical qubits $\gamma_1, \gamma_2 \in H_1(C)$ can be written as
\begin{align}
    \overline{CZ}(\gamma_1,\gamma_2)
    = \exp\left(i\pi \frac{I-\overline{Z}(\gamma_1)}{2} \frac{I-\overline{Z}(\gamma_2)}{2}\right).
\end{align}

Let us evaluate the exponential on the right-hand side as 
\begin{align}
    \frac{I-\overline{Z}(\gamma_1)}{2}\frac{I-\overline{Z}(\gamma_2)}{2}
    &= \sum_{i,j}\gamma_1^i \gamma_2^j z_i z_j + 2 \Delta',
\end{align}
where $\Delta'$ contains terms with three or more $z_i$ factors, and gives $\exp(2i\pi\Delta')=1$ since it has integer eigenvalues.
Therefore, the logical CZ gate is simplified as
\begin{align}
    \overline{CZ}(\gamma_1,\gamma_2)
    = \exp\left(i\pi \sum_{i,j}\gamma_1^i \gamma_2^j z_i z_j\right).
\end{align}
We can further decompose this operator into physical CZ gates.
Let $\mathrm{supp}(\gamma_1)=\{i_1,\ldots,i_m\}$ and $\mathrm{supp}(\gamma_2)=\{j_1,\ldots,j_n\}$ denote the supports of the two logical qubits.
Then the logical CZ gate decomposes into a product of physical CZ gates for any pair of qubits in the supports as
\begin{align}
    \overline{CZ}(\gamma_1,\gamma_2)
    = \prod_{k=1}^{m} \prod_{l=1}^{n} CZ_{i_k,j_l}.
\end{align}
For consistency, we define the CZ gate with identical control and target qubits as the $Z$ gate itself:
\begin{align}
    CZ_{i,i} = \exp(i\pi z_i z_i) = \exp(i\pi z_i) = Z_i.
\end{align}

Note that the exponent can be compactly expressed as
\begin{align}
    \sum_{i,j}\gamma_1^i\gamma_2^j z_i z_j = a(\gamma_1)a(\gamma_2),
\end{align}
and thus the logical CZ gate takes a particularly simple gauge field form:~\footnote{
    In the $U(1)$ gauge theory correspondence, this logical CZ gate may be expressed as
    \begin{align}
        \overline{CZ}(\gamma_1,\gamma_2) \leftrightarrow \exp(i \oint_{\gamma_1} dx\, a(x) \oint_{\gamma_2} dy\, a(y)),
    \end{align}
    where $\gamma_1$ and $\gamma_2$ are nontrivial loops of the manifold where the theory is defined.
}:
\begin{align}
    \label{eq:gauge_field_expression_of_logical_cz_gate}
    \overline{CZ}(\gamma_1,\gamma_2) = \exp(i\pi a(\gamma_1)a(\gamma_2)).
\end{align}

With the above gauge field expression \Cref{eq:gauge_field_expression_of_logical_cz_gate} of the logical CZ gate, we can easily verify that the logical CZ gate depends only on the homology classes of $\gamma_1$ and $\gamma_2$.
The detailed proof is provided in \Cref{app:cohomology_dependence_of_logical_gates}, but the essential idea is as follows.
 Decompose $\gamma_i = \gamma'_i + \del f_i$ with some face $f_i$ for $i=1,2$.
Then
\begin{multline}
    \exp(i\pi a(\gamma_1)a(\gamma_2))
    = \exp(i\pi a(\gamma'_1)a(\gamma'_2)) \times  \\ 
      \exp(i\pi a(\gamma'_1)a(\del f_2))
        \exp(i\pi a(\del f_1)a(\gamma'_2))
        \exp(i\pi a(\del f_1)a(\del f_2)).
\end{multline}
All terms in the second line give logical identities, since they correspond to products of controlled-$Z$ stabilizer gates, which act trivially on the code space.
Thus, we conclude that
\begin{align}
    \overline{CZ}(\gamma_1,\gamma_2) \sim \overline{CZ}(\gamma'_1, \gamma'_2),
\end{align} 
demonstrating that the logical CZ gate depends only on the homology classes of $\gamma_1$ and $\gamma_2$.
This confirms that the above decomposition into physical CZ gates is a valid logical CZ gate.

The above discussion extends naturally to the multicontrolled-$Z$ gate with $m-1$ control qubits and one target qubit, which we denote by $C^{m-1}Z$. 
It can be formally defined as
\begin{align}
    C^mZ_{i_1,\ldots,i_m}
    &= I_{i_m} (I_{i_1}\cdots I_{i_{m-1}} - P^+_{i_1}\cdots P^+_{i_{m-1}})
    + Z_{i_m} P^+_{i_1}\cdots P^+_{i_{m-1}}\\
    &= \exp(i\pi \frac{I-Z_{i_1}}{2}\cdots\frac{I-Z_{i_m}}{2}).
\end{align}
Accordingly, the logical $C^{m-1}Z$ gate acting on $m$ logical qubits $\gamma_1,\ldots,\gamma_m\in H_1(C)$ is given by
\begin{align}
    \overline{C^{m-1}Z}(\gamma_1,\ldots,\gamma_m)
    = \exp(i\pi \frac{I-\overline{Z}(\gamma_1)}{2}\cdots \frac{I-\overline{Z}(\gamma_m)}{2}).
\end{align}
To obtain the physical decomposition, we expand the product in the exponent modulo $2$ as
\begin{align}
    \label{eq:logical_cmz_gate_decomposition}
    \frac{I-\overline{Z}(\gamma_{i_1})}{2}\cdots\frac{I-\overline{Z}(\gamma_{i_m})}{2}
    = \sum_{j_1,\ldots,j_m} \gamma_{i_1}^{j_1}\cdots\gamma_{i_m}^{j_m} z_{j_1}\cdots z_{j_m}
    \quad (\text{mod } 2).
\end{align}
Hence, the logical $C^{m-1}Z$ decomposes as follows:
\begin{align}
    \overline{C^{m-1}Z}(\gamma_1,\ldots,\gamma_m)
    = \prod_{k=1}^{\alpha_1}\cdots\prod_{l=1}^{\alpha_m} C^{m-1}Z_{s^1_k,\ldots,s^m_l},
\end{align}
where we defined the supports of each logical qubit as $\mathrm{supp}(\gamma_n) = \{s^n_1,\ldots,s^n_{\alpha_n}\}$ for $n=1,\ldots,m$.
This shows that the logical $C^{m-1}Z$ gate is realized as the product of physical $C^{m-1}Z$ gates that act on all possible combinations of the supports of $\gamma_1,\ldots,\gamma_m$.
For consistency, we recursively define the degenerate cases of $C^{m-1}Z$:
if two or more control qubits coincide, the operation reduces to one with fewer control qubits,
and if control and target coincide, it reduces to the $Z$ gate.
For example, $C^3Z_{i,i,j,k} = C^2Z_{i,j,k}, C^2Z_{i,i,i} = CZ_{i,i} = Z_i$.
This construction agrees with the round-robin construction, given in Ref.~\cite{Yoder:2016bys}.

Note that the exponential term in \Cref{eq:logical_cmz_gate_decomposition} can again be expressed compactly using gauge fields as
\begin{align}
    \sum_{i_1,\ldots,i_m}\gamma_{i_1}^{j_1}\cdots\gamma_{i_m}^{j_m} z_{j_1}\cdots z_{j_m}
    = a(\gamma_1)\cdots a(\gamma_m).
\end{align}
Then, the logical multicontrolled-$Z$ gate takes the compact gauge field form:
\begin{align}
    \overline{C^{m-1}Z}(\gamma_1,\ldots,\gamma_m)
    = \exp(i\pi a(\gamma_1)\cdots a(\gamma_m)).
\end{align}
This expression is consistent with previous work~\cite{breuckmann2024cupsgatesicohomology,hsin2024classifyinglogicalgatesquantum,zhu2025topologicaltheoryqldpcnonclifford}, 
where logical $C^mZ$ gates were constructed using the cup product.

The operation of these multicontrolled-$Z$ gates on the code space is well defined, as is the same as the CZ gate case:
\begin{align}
    \overline{C^{m-1}Z}(\gamma_1,\ldots,\gamma_m)
    \sim \overline{C^{m-1}Z}(\gamma'_1,\ldots,\gamma'_{m}),
\end{align}
where $\gamma_i, \gamma'_i \in \ker \del_1$ for $i=1,\ldots,m$ satisfy $\gamma_i = \gamma'_i + \del_2 f_i$ with some $f_i \in C_2$.  
The detailed proof is provided in \Cref{app:cohomology_dependence_of_logical_gates}.
This confirms that the above decomposition into physical multicontrolled-$Z$ gates is a valid logical $C^{m-1}Z$ gate.

\subsection{T gate}

So far, we have shown that the logical $S$, Hadamard, and $CZ$ gates can be constructed from physical gates using the gauge field formalism.
All of these are Clifford gates, and their decompositions can be achieved using only physical Clifford operations, as demonstrated in the previous subsections.
In this subsection, we turn to the logical $T$ gate, which is outside the Clifford group.
As one might expect, its logical implementation requires physical non-Clifford gates such as the $T$ gate, the controlled-$S$ gate, and the controlled-controlled-$Z$ ($CCZ$) gate.
Although the resulting decomposition is somewhat involved, we will show that the gauge field formalism allows it to be expressed in a compact form.

The physical $T$ gate is defined as
\begin{align}
    T = 
    \mqty(1 & 0 \\ 0 & e^{i\pi/4})
    = \exp(i\frac{\pi}{4}\frac{I - Z}{2}).
\end{align}
Following the same prescription as before, the logical $T$ gate acting on a single logical qubit $\gamma \in H_1(C)$ is given by
\begin{align}
    \overline{T}(\gamma)
    = \exp(i\frac{\pi}{4}\frac{I - \overline{Z}(\gamma)}{2}).
\end{align}

Expanding the exponential on the right-hand side modulo 2, we have
\begin{align}
    \frac{1}{4}\frac{I - \overline{Z}(\gamma)}{2}
    = \frac{1}{4}\sum_i \gamma^i z_i
      - \frac{1}{2}\sum_{i<j}\gamma^i\gamma^j z_i z_j
      + \sum_{i<j<k}\gamma^i\gamma^j\gamma^k z_i z_j z_k \quad (\text{mod } 2).
\end{align}
Each term corresponds, respectively, to single-, two-, and three-qubit interactions.
Exponentiating these terms yields the complete logical $T$ gate decomposition:
\begin{align}
    \overline{T}(\gamma)
    &= \prod_i \exp(i\frac{\pi}{4}\gamma^i z_i)
       \prod_{i<j}\exp(-i\frac{\pi}{2}\gamma^i\gamma^j z_i z_j)
       \prod_{i<j<k}\exp(i\pi\gamma^i\gamma^j\gamma^k z_i z_j z_k)\\
    &= \prod_{k=1}^m T_{i_k}
       \prod_{1\leq k_1 < k_2 \leq m} CS_{k_1,k_2}^\dagger
       \prod_{1\leq k_1 < k_2 < k_3 \leq m} CCZ_{k_1,k_2,k_3},
\end{align}
where $T_i$ is the physical $T$ gate acting on the $i$th qubit, $CS_{i,j}$ is the controlled-$S$ gate with the control qubit $i$ and the target $j$, and $CCZ_{i,j,k}$ is the controlled-controlled-$Z$ gate with control qubits $i,j$ and the target $k$.
We defined the support of $\gamma$ as $\mathrm{supp}(\gamma) = \{i_1,\ldots,i_m\}$ with $i_k < i_{k+1}$.
We find that the logical $T$ gate is implemented by applying transversal $T$ gates on the support of $\gamma$, controlled-$S$ gates between every pair within that support (with the dagger indicating the inverse operation), and $CCZ$ gates among every triplet within that support.
We define the controlled-$S$ gate with identical control and target as the $S$ gate itself as $CS_{i,i} = S_i$, similarly to the CZ gate case.

We can express this result more elegantly in terms of the gauge field.
Using the relations in \Cref{app:useful_formulae}, we obtain~\footnote{
The counterpart to this logical $T$ gate in the $U(1)$ gauge theory correspondence may be expressed as
\begin{align}
    \overline{T}(\gamma) \leftrightarrow \exp(
        \frac{i}{4} \biggl(\oint_\gamma dx\, a(x)\biggr)^3
        -\frac{3i}{4} \biggl(\oint_\gamma dx\, a(x)\biggr)^2
        + \frac{i}{2} \oint_\gamma dx\, a(x)),
\end{align}
which involves the cube and square of the loop integral of the gauge field.
}
\begin{align}
    \overline{T}(\gamma) = 
    \exp(i\frac{\pi}{4}\frac{I - \overline{Z}(\gamma)}{2})
    = \exp(i\pi \left(\frac{1}{2}a(\gamma)^3 - \frac{3}{4}a(\gamma)^2 + \frac{1}{2}a(\gamma)\right)).
\end{align}
This expression highlights that the $T$ gate introduces a cubic contribution of the gauge field $a(\gamma)$, which elevates it beyond the Clifford hierarchy level of the logical $S$ and $CZ$ gates.

In this expression, the exponent consists of a linear term $a(\gamma)$, a quadratic term $a(\gamma)^2$, and a cubic term $a(\gamma)^3$.
The linear term $a(\gamma)$ corresponds to single-qubit phase rotations supported on $\gamma$.
The quadratic term $a(\gamma)^2$ should be understood as a compact representation that bundles all two-body diagonal phase interactions among qubits within the support of $\gamma$, rather than the square of an independent Wilson line.
Similarly, the cubic term $a(\gamma)^3$ represents a collective expression for all three-body diagonal phase interactions on the same support, and is not the cube of an independent Wilson line.

In the physical decomposition of $\overline{T}(\gamma)$, the linear, quadratic, and cubic terms give rise to single-qubit $T$ gates, pairwise $CS^\dagger$ gates, and $CCZ$ gates on all triplets within the support of $\gamma$, respectively.
Taken together, these contributions define a consistent logical operator and ensure the gauge invariance of the Wilson line-like operator.

The logical operation of this logical $T$ gate depends only on the homology class of $\gamma$, similar to the previous logical gates.
The proof is provided in \Cref{app:cohomology_dependence_of_logical_gates}, which is hard to show with only its physical gate decomposition and the gauge field expression greatly simplifies the analysis.
Finally, for $\gamma_1, \gamma_2 \in \ker \del_1$ satisfying $\gamma_1 = \gamma_2 + \del_2 f$ with some $f \in C_2$, we have
\begin{align}
    \overline{T}(\gamma_1) \sim \overline{T}(\gamma_2),
\end{align}
which confirms that the above decomposition into physical gates is a valid logical $T$ gate.

\subsection{Discussion}
We have constructed a broad class of logical gates within the gauge field formalism in the previous subsections. 
As we have seen, the gauge field expressions provide compact and insightful representations of these logical gates, facilitating the analysis of their properties, such as homology dependence.

How can our method be extended to other logical gates?
Clearly, it can apply to any diagonal gate whose form can be written in terms of the Pauli $Z$ operator, given by
\begin{align}
    \mathcal{O}_k(\gamma_1,\ldots,\gamma_m)
    \coloneqq
    \exp(i\frac{\pi}{2^{k-1}}
        \frac{I-\overline{Z}(\gamma_1)}{2}\cdots 
        \frac{I-\overline{Z}(\gamma_m)}{2}
        )
\end{align}
with positive integers $k,m$.
The cases $m=1$ and $k=1,2,3$ reproduce the logical $Z$, $S$, and $T$ gates, while $m\geq 2$ gives their controlled versions.
In parallel, we can construct diagonal gates in the $+$ basis, expressed in terms of the Pauli $X$ operator with the magnetic gauge field $b$.
Furthermore, like the logical Hadamard gate, our method can handle logical gates formed by alternating layers of diagonal gates in the $Z$ and $X$ bases.

However, these gates do not exhaust all possible logical operations, and extending our method to more general nondiagonal gates faces several challenges.
For example, the logical CNOT may be obtained from the relation between CNOT, CZ, and Hadamard gates:
\begin{align}
    CNOT_{i,j} = H_j\, CZ_{i,j}\, H_j,
\end{align}
although this approach leads to a deep circuit that involves many physical Hadamard and $S$ gates.
Alternatively, we may start from the exponential form
\begin{align}
    CNOT_{i,j}
    = \exp\!\left(i\pi \frac{I-Z_i}{2}\frac{I-X_j}{2}\right),
\end{align}
which suggests the logical operator
\begin{align}
    \overline{CNOT}(\gamma_1,\tilde{\gamma}_2)
    = \exp\!\left(i\pi
    \frac{I-\overline{Z}(\gamma_1)}{2}
    \frac{I-\overline{X}(\tilde{\gamma}_2)}{2}\right),
\end{align}
where $\gamma_1\in H_1(C)$ and $\tilde{\gamma}_2\in H^1(C)$ represent the logical control and the target qubits.  
If the supports of $\gamma_1$ and $\tilde{\gamma}_2$ do not overlap, the corresponding electric and magnetic fields commute, allowing a decomposition similar to the diagonal case.
In general, however, they may overlap, and the resulting noncommuting gauge fields obstruct a direct extension of our techniques.
Developing methods for efficiently handling such noncommuting fields remains an open direction.

We now comment on the homology dependence of our logical gates.  
Although one might imagine it to be a special feature of the specific gates studied in this work, this property holds in general for any logical operator constructed through our first and second prescriptions.
Let a logical operator be expressed as a polynomial in logical Pauli operators:
\begin{multline}
    \overline{U}(\overline{X}(\tilde{\gamma}_1),\ldots,
    \overline{X}(\tilde{\gamma}_k),
    \overline{Z}(\gamma_1),\ldots,\overline{Z}(\gamma_k))
    \\
    = \sum_{i_1,\ldots,i_{2k}}
    c_{i_1,\ldots,i_{2k}}
    \,
    \overline{X}(\tilde{\gamma}_1)^{i_1}\cdots 
    \overline{X}(\tilde{\gamma}_k)^{i_k}
    \overline{Z}(\gamma_1)^{i_{k+1}}\cdots
    \overline{Z}(\gamma_k)^{i_{2k}}.
\end{multline}
Write $\gamma_m=\gamma'_m+\partial_2 f_m$ and 
$\tilde{\gamma}_n=\tilde{\gamma}'_n+\partial_1^T v_n$ for some faces $f_m$ and vertices $v_n$ with another set of 1-chains $\gamma'_m$ and $\tilde{\gamma}'_n$. 
Substituting these expressions, all terms involving boundaries reduce to stabilizers, and hence act trivially on the code space.  
Thus, each monomial is logically equivalent to the one with boundaries removed, and summing over all monomials yields the following:
\begin{align}
    \overline{U}(\overline{X}(\tilde{\gamma}_1),\ldots,\overline{Z}(\gamma_k))
    \sim
    \overline{U}(\overline{X}(\tilde{\gamma}'_1),\ldots,\overline{Z}(\gamma'_k)),
\end{align}
demonstrating that the logical operation of $\overline{U}$ depends only on the (co)homology classes of its arguments.  
The gauge field calculations presented earlier simply make this dependence explicit, whereas establishing it directly from the physical gate decompositions would be highly nontrivial.

Finally, let us remark on multiqubit gates acting across different copies of a code.  
If we consider several CSS codes simultaneously, the natural setting of the whole Hilbert space is the direct sum of the individual ones.
The gauge fields can extend straightforwardly to this direct sum, and the logical gates can be constructed in the same manner as in the single-code case.  
In particular, nondiagonal gates between two code blocks pose no difficulty, since the gauge fields associated with different code copies act on nonoverlapping supports and therefore commute.  
For example, the logical CNOT between two code copies becomes a transversal physical CNOT applied between all qubits in the supports of the control and target logical qubits, consistent with the known results~\cite{Gottesman:1997zz,hsin2024classifyinglogicalgatesquantum}.

\section{Conclusion and Outlook}
\label{sec:conclusion}

In this work, we developed a systematic method for constructing logical operations in general quantum CSS codes using the gauge field formalism.  
By introducing electric and magnetic gauge fields $a$ and $b$ as operator-valued cochains on the underlying chain complex, the formalism establishes a direct algebraic correspondence between the structure of CSS codes and concepts from lattice gauge theory.  
Within this framework, we expressed a broad class of logical Clifford and non-Clifford gates—including the $S$, Hadamard, $CZ$, multicontrolled-$Z$, and $T$ gates—as exponential of polynomial functions of these gauge fields.  
This expression provides transparent decompositions into physical gates and shows that the logical operation of these gates depends only on the (co)homology classes of the associated 1-chains, guaranteeing well-definedness at the logical level.

A natural direction for future work is to take into account fault tolerance and circuit depth.  
The constructions presented here focus solely on expressing logical gates in terms of physical operators, without imposing noise-resilience or architectural constraints.  
Practical implementations require logical gates that remain robust under realistic noise models, and developing a fault-tolerant formulation directly within the gauge field language may illuminate how transversality and other protection mechanisms emerge from code structure.
Incorporating recent progress on transversal logical gates, such as Clifford gates for self-dual CSS codes~\cite{Tansuwannont:2025riy} and diagonal gates for general stabilizer codes~\cite{Anderson:2014voa}, we may clarify the algebraic origins of these constraints from the gauge field viewpoint.
Furthermore, many fault-tolerant protocols, such as gate teleportation~\cite{Gottesman:1999tea}, also rely on ancilla qubits and measurements.  
Integrating such measurement-based or ancilla-assisted procedures with the gauge field framework likely requires extending the formalism to density matrices and quantum channels, an important step toward connecting this approach with realistic fault-tolerant architectures.

Another line of extensions concerns stabilizer codes beyond the CSS codes.  
Although general stabilizer codes do not separate clearly into $X$- and $Z$-type operators, they still admit descriptions using length-2 chain complexes, suggesting that the gauge field formalism may generalize after appropriately reformulating the correspondence between qubits, stabilizers, and chain-complex data.  
Moreover, subsystem codes~\cite{Kribs:2004dqj,Kribs:2006doz} offer further opportunities, as their intrinsic redundant (``gauge") degrees of freedom naturally invite a gauge-theoretic interpretation.  
Identifying how these additional structures manifest in the formalism may deepen the conceptual connection between quantum error correction codes and lattice gauge theory.

The framework may also be extended to codes defined on higher-dimensional chain complexes.  
Our analysis focused on code-capacity noise, where errors occur only on physical qubits; however, fault-tolerant architectures must handle measurement errors as well. 
Accounting for these errors leads to longer chain complexes, such as in single-shot error-correcting codes represented by length-4 complexes with qubits on 2-chains and measurement-error syndromes on 0- and 4-chains~\cite{Campbell:2019cyo}.  
In such settings, the associated gauge theories reside on 4D hypergraphs rather than 2D ones, and exploring their properties may reveal new insights into fault tolerance against data and measurement errors.

In addition, there are higher-dimensional generalizations of CSS codes, such as higher-dimensional toric codes~\cite{Hastings:2016jmy}, which are defined on higher-dimensional chain complexes.
In such models, physical qubits are generally placed on $p$-cells for some $p$, while the $X$-type and $Z$-type stabilizers are supported on the adjacent $(p-1)$- and $(p+1)$-cells, respectively.
If one restricts attention only to the algebraic data of the resulting length-2 subcomplex,
\begin{align}
C_{p+1} \xrightarrow{\partial_{p+1}} C_p \xrightarrow{\partial_p} C_{p-1},
\end{align}
then part of our present construction can still be applied at the level of this CSS subcomplex.
However, such a reduction does not faithfully retain the original higher-dimensional geometric locality of the model.
For a faithful local description of higher-dimensional toric codes within the gauge-theoretic language, it is therefore more natural to generalize the present formalism to higher-form gauge fields rather than to regard them simply as ordinary 1-form gauge theories on hypergraphs.

A broader question concerns the range of quantum error-correcting codes that can be formulated within this gauge-theoretic language.  
For CSS codes on the level-$N$ qudits, the associated gauge theories are $\mathbb{Z}_N$ lattice gauge theories~\cite{hsin2024classifyinglogicalgatesquantum}, which approach $U(1)$ gauge theories in the large-$N$ limit and are the foundation of numerical studies of lattice gauge theories~\cite{Creutz:1982dn,Creutz:1983ev}.
However, physically relevant gauge groups are often non-Abelian; the Standard Model in particle physics is based on the gauge group $SU(3)\times SU(2)\times U(1)$~\cite{Weinberg:1967tq}.  
Finite subgroups of $SU(2)$ and $SU(3)$, such as binary polyhedral groups or $\mathrm{PSL}(2,\mathbb{Z}_p)$, have long been used as discrete approximations of lattice gauge theories with these non-Abelian gauge groups~\cite{Petcher:1980cq,Bhanot:1981xp,Bhanot:1981pj,Kubischta:2023nlb,Denys:2023syu,Carena:2024dzu,Assi:2024pdn,Lin:2024utr,Kurkcuoglu:2025gik,Perez:2025cxl}.  
Constructing quantum codes associated with such finite non-Abelian groups and formulating their gauge field descriptions would represent a significant step toward developing new classes of quantum error-correcting codes.

In summary, the gauge field formalism offers a compact and flexible language for describing logical operations in CSS codes.  
Extending it to more general stabilizer and subsystem codes, to higher-dimensional chain complexes relevant for realistic noise models, and to non-Abelian groups may uncover deeper connections between quantum error correction, algebraic topology, and quantum field theory.

\section*{Acknowledgments}

The authors thank people in the QEC team of the Fujii lab and in the Center for Quantum Information and Quantum Biology at the University of Osaka for valuable discussions.
This work was supported by MEXT Quantum Leap Flagship Program (MEXT Q-LEAP) Grant No. JPMXS0120319794, JST COI-NEXT Grant No. JPMJPF2014, JST Moonshot R\&D Grant Nos. JPMJMS2061 and JPMJMS256E, and JST CREST Grant No. JPMJCR24I3.

\appendix
\section*{Appendix}
\numberwithin{equation}{section}

\section{Useful Formulae and Theorems}
\label[appendix]{app:useful_formulae}
In this appendix, we summarize some useful formulae and theorems used in the main text.

\subsection{Commutation relation}
\label[appendix]{app:commutation_relation}
Here, we derive the commutation relation between the gauge field $a,b$ and the Pauli operator $X,Z$.
First, we consider the commutation relation between a single-qubit Pauli $X_j$ and $z_i$, given by
\begin{align}
    X_j z_i X_j &= z_i + \delta_{ij} - 2\delta_{ij}z_i 
    = \exp(i\pi \inp{e_i}{e_j}) z_i + \frac{1 - \exp(i\pi \inp{e_i}{e_j})}{2},
\end{align}
where we used the definition of the inner product $\inp{e_i}{e_j} = \delta_{ij}$ and the fact that for a binary variable $\alpha = 0,1$,
\begin{align}
    \label{eq:useful_formula_for_exp_of_binary_variable}
    \exp(i\pi \alpha) = 1 - 2\alpha.
\end{align}
Based on this commutation relation, we can show that the commutation relation between the Pauli $X$ operator $\overline{X}(\gamma_2) = \prod_j X_j^{\gamma_2^j}$ and $z_i$ is given by
\begin{align}
    \overline{X}(\gamma_2) z_i \overline{X}(\gamma_2) 
    &= \exp(i\pi \inp{e_i}{\gamma_2})z_i + \frac{1 - \exp(i\pi \inp{e_i}{\gamma_2})}{2}\\
    &= (1 - 2\inp{e_i}{\gamma_2})z_i + \inp{e_i}{\gamma_2}\\
    & = a(e_i) - 2 a(e_i \cap \gamma_2) + \inp{e_i}{\gamma_2},
    \label{eq:commutation_relation_between_a_and_X}
\end{align}
where we again used \Cref{eq:useful_formula_for_exp_of_binary_variable} on the second line and the definition of the intersection $\gamma_2 \cap e_i = \inp{\gamma_2}{e_i} e_i$ on the last line.
Multiplying $\gamma_1^i$ by \Cref{eq:commutation_relation_between_a_and_X} and taking summation over $i$, we find that the commutation relation between $a(\gamma_1) = \sum_i \gamma_1^i z_i$ and $\overline{X}(\gamma_2)$ is given by
\begin{align}
    \overline{X}(\gamma_2) a(\gamma_1) \overline{X}(\gamma_2) &= a(\gamma_1) - 2a(\gamma_1 \cap \gamma_2) + \inp{\gamma_1}{\gamma_2}.
\end{align}
Multiplying the Hadamard operator $H = \prod_i H_i$ from both sides and using the relation $H z_i H = x_i$, we obtain the commutation relation between the magnetic gauge field $b(\gamma_1)$ and $\overline{Z}(\gamma_2)$ as
\begin{align}
    \overline{Z}(\gamma_2) b(\gamma_1) \overline{Z}(\gamma_2) &= b(\gamma_1) - 2b(\gamma_1 \cap \gamma_2) + \inp{\gamma_1}{\gamma_2}.
\end{align} 
Using the above relations, we can also derive the familiar commutation relation between the Pauli $X$ and $Z$ operators as 
\begin{align}
    X(\gamma_1) Z(\gamma_2) X(\gamma_1) &= \exp(i\pi X(\gamma_1) a(\gamma_2) X(\gamma_1)) \\
    &= \exp(i\pi (a(\gamma_2) - 2 a(\gamma_1 \cap \gamma_2) + \inp{\gamma_1}{\gamma_2})) \\
    &= \exp(i\pi \inp{\gamma_1}{\gamma_2}) Z(\gamma_2),
\end{align}
where we used the fact that $\exp(-2i\pi a(\gamma \cap \gamma_2)) = I$ since $a(\gamma \cap \gamma_2)$ is a diagonal matrix with integer eigenvalues, commuting with $a(\gamma_2)$.

\subsection{Expansions of powers of electric gauge field}
In this section, we summarize useful formulae by expanding the powers of the electric gauge field $a(\gamma) = \sum_i \gamma^i z_i$.
Using the fact that $(\gamma^i)^2 = \gamma^i$ and $z_i^2 = z_i$, we find the following relations:
\begin{subequations}
    \begin{align}
        a(\gamma) &= \sum_i \gamma^i z_i \\
        a(\gamma)^2 &= \sum_i \gamma^i z_i + 2\sum_{i<j} \gamma^i \gamma^j z_i z_j \\
        a(\gamma)^3 &=  \sum_i \gamma^i z_i + 6\sum_{i<j} \gamma^i \gamma^j z_i z_j + 6\sum_{i<j<k} \gamma^i \gamma^j \gamma^k z_i z_j z_k.
    \end{align}
\end{subequations}
More generally, the $m$th power of $a(\gamma)$ can be expanded as
\begin{align}
    \label{eq:expansion_of_a_gamma_power}
    a(\gamma)^m = \sum_{r=1}^m r! S(m,r) \sum_{i_1 < \cdots < i_r} \gamma^{i_1} \cdots \gamma^{i_r} z_{i_1} \cdots z_{i_r},
\end{align}
where $S(m,r)$ is the Stirling number of the second kind, which counts the number of ways to partition a set of $m$ elements into $r$ nonempty subsets.

Solving for the terms on the right-hand side, we find that
\begin{subequations}
    \label{eq:useful_formulae_a_gamma}
    \begin{align}
        \sum_i \gamma^i z_i & = a(\gamma), \\
        \label{eq:useful_formula_a_gamma_squared}
        \sum_{i<j} \gamma^i \gamma^j z_i z _j & = \frac{1}{2} a(\gamma)^2 - \frac{1}{2} a(\gamma), \\
        \sum_{i<j<k} \gamma^i \gamma^j \gamma^k z_i z_j z_k & = \frac{1}{6} a(\gamma)^3 - \frac{1}{2} a(\gamma)^2 + \frac{1}{3} a(\gamma).
    \end{align}
\end{subequations}
Based on \Cref{eq:expansion_of_a_gamma_power}, we can also derive the following general formula:
\begin{align}
    \sum_{i_1 < \cdots < i_r} \gamma^{i_1} \cdots \gamma^{i_r} z_{i_1} \cdots z_{i_r}
    = \frac{1}{r!} \sum_{m=1}^r (-1)^{m+r} s(r,m) a(\gamma)^m,
\end{align}
where $s(r,m)$ is the Stirling number of the first kind, which counts the number of permutations of $r$ elements with exactly $m$ disjoint cycles.

\subsection{Necessary and sufficient condition for identity operator in qubit system}
\label[appendix]{sec:IdentityCondition}
In this section, we show that a necessary and sufficient condition for an $n$-qubit operator $O$ to be the identity operator is given by the commutativity with all the Pauli operators.

\begin{theorem}
\label{thm:IdentityCondition}
An $n$-qubit operator $O$ is the identity operator with a constant factor if and only if it commutes with all the Pauli operators.
That is,
\begin{align}
O = c\, I_{2^n} \text{ with } c \in \mathbb{C} \iff {}^\forall n \in \mathbb{N},\,[X_n, O] = [Z_n, O] = 0.
\end{align}
\end{theorem}

\begin{proof}
The sufficiency ($\Rightarrow$) is obvious.
Let us show the necessity ($\Leftarrow$) by mathematical induction on $n$.

\noindent \underline{Base case ($n=1$)}

 Consider a $1$-qubit operator $O$.
For any $1$-qubit operator $O$, we can express it as
\begin{align}
O = a_0 I_1 + a_1 X_1 + a_2 Y_1 + a_3 Z_1,
\end{align}
where $a_i \in \mathbb{C}$.
Applying condition $[X_1,O]=0$, or equivalently, $O = X_1 O X_1$, we get
\begin{align}
 a_0 I_1 + a_1 X_1 + a_2 Y_1 + a_3 Z_1 = a_0 I_1 + a_1 X_1 - a_2 Y_1 - a_3 Z_1,
\end{align}
which yields
\begin{align}
a_2 Y_1 + a_3 Z_1 =0.
\end{align}
The solution of this equation is given by $a_2 = 0$ and $a_3 = 0$.
Similarly, we can show that $a_1 = 0$ by applying the condition $[Z_1,O]=0$.
Therefore, we have shown that $O = a_0 I_1$.
Thus, the base case is proved.

\noindent \underline{Inductive step}

Assume that the theorem holds for $n$-qubit operators and consider an $(n+1)$-qubit operator $O$ that commutes with all the $(n+1)$-qubit Pauli operators.
Then, we first consider the commutativity with the Pauli operators on the $n+1$th qubit $X_{n+1},Z_{n+1}$.
We can express $O$ as
\begin{align}
    O = \mqty( O_1 & O_2 \\ O_3 & O_4),
\end{align}
where $O_1,O_2,O_3,O_4$ are $2^n \times 2^n$ matrices.
On the other hand, $X_{n+1}= I_{2^n} \otimes X$ and $Z_{n+1}= I_{2^n} \otimes Z$ are expressed as
\begin{align}
X_{n+1} = \mqty( 0 & I_{2^n} \\ I_{2^n} & 0), \quad Z_{n+1} = \mqty( I_{2^n} & 0 \\ 0 & -I_{2^n}).
\end{align}
Then, the conditions $[X_{n+1},O]=[Z_{n+1},O]=0$ are rewritten as
\begin{align}
X_{n+1} O X_{n+1} = \mqty( O_4 & O_3 \\ O_2 & O_1) = \mqty (O_1 & O_2 \\ O_3 & O_4) = O,\\
Z_{n+1} O Z_{n+1} = \mqty( O_1 & -O_2 \\ -O_3 & O_4) = \mqty (O_1 & O_2 \\ O_3 & O_4) = O,
\end{align}
which yields 
\begin{align}
O_1 = O_4, \quad O_2 = O_3 = 0.
\end{align}
Therefore, we can express $O$ as
\begin{align}
O = \mqty( O_1 & 0 \\ 0 & O_1) = O_1 \otimes I_2.
\end{align}
At this point, $O_1$ is a $2^n \times 2^n$ operator and must commute with the rest of the $n$-qubit Pauli operators. 
Then, we can apply the inductive hypothesis to $O_1$ and conclude that $O_1 = c I_{2^n}$ with some constant $c \in \mathbb{C}$.
Thus, we have shown that $O = c I_{2^{n}} \otimes I_2 = c I_{2^{n+1}}$.
This completes the proof of the theorem.
\end{proof}

From this theorem we can easily see the relation between the type of operator $O$ and the constant factor $c$ as 
\begin{enumerate}
    \item $O$ : Hermitian $\iff$  $c \in \mathbb{R}$
    \item $O$ : unitary $\iff$ $c = e^{i\theta}$ with $\theta \in \mathbb{R}$
    \item $O$ : Hermitian and unitary $\iff$ $c = \pm 1$
\end{enumerate}

\section{(Co)homology Dependence of Logical Gates}
\label[appendix]{app:cohomology_dependence_of_logical_gates}
In this appendix, we give detailed proof that the logical gates constructed in \Cref{sec:logical_gate_and_gauge_field_formalism} depend only on the homology or cohomology class of their arguments.

Before moving onto examples, we introduce here the notion of logical equivalence between logical operators, as was briefly mentioned in the main text.
Let us consider two logical operators $\mathcal{O}_1$ and $\mathcal{O}_2$ and assume that they are unitary on the full Hilbert space of physical qubits, as usual in quantum error correction.
When they have the same operation (representation matrix) on the code space, these gates should be regarded as the same logical operator, even if they consist of different physical gates.
It is helpful to introduce the equivalence relation $\sim$ between these operators as 
\begin{align}
    \label{eq:equivalence_relation_for_logical_operator}
    \mathcal{O}_1 \sim \mathcal{O}_2
    \overset{\text{def}}{\iff} 
    \bra{\psi} \mathcal{O}_1 \ket{\phi} = \bra{\psi} \mathcal{O}_2 \ket{\phi},
\end{align}
for any arbitrary elements of the code space $\ket{\psi}$ and $\ket{\phi}$.
This definition is equivalent to
\begin{align}
    \mathcal{O}_1 \sim \mathcal{O}_2 \iff \mathcal{O}_1\mathcal{O}_2^{-1} \sim I.
\end{align}
This equivalence can be shown as follows.
 Consider a basis for the code space $\{\ket{\psi}_i\}$.
Then, we have
\begin{align*}
    &\bra{\psi_i} \mathcal{O}_1 \ket{\psi_j} = \bra{\psi_i} \mathcal{O}_2 \ket{\psi_j} 
    \\
    & \iff \sum_j \bra{\psi_i} \mathcal{O}_1 \ket{\psi_j} \bra{\psi_j}\mathcal{O}_2^{-1} \ket{\psi_k} = \sum_j \bra{\psi_i} \mathcal{O}_2 \ket{\psi_j} \bra{\psi_j} \mathcal{O}_2^{-1} \ket{\psi_k} \\
    & \iff \bra{\psi_i}\mathcal{O}_1 P_{\mathcal{C}}\mathcal{O}_2^{-1} \ket{\psi_k} = \bra{\psi_i} \mathcal{O}_2 P_{\mathcal{C}} \mathcal{O}_2^{-1} \ket{\psi_k} \\ 
    & \iff \bra{\psi_i} \mathcal{O}_1\mathcal{O}_2^{-1} \ket{\psi_k} = \bra{\psi_i} I \ket{\psi_k} \\
    & \iff \mathcal{O}_1\mathcal{O}_2^{-1} \sim I,
\end{align*}
where $P_{\mathcal{C}}\coloneqq \sum_j \ket{\psi_j}\bra{\psi_j}$ is the projector onto the code space, and we used the fact that $P_{\mathcal{C}} \mathcal{O}_2^{-1} \ket{\psi_j} = \mathcal{O}_2^{-1} \ket{\psi_j}$.
Note that the inverse $\mathcal{O}_2^{-1}$ of a logical operator $\mathcal{O}_2$ is also a logical operator, since it is unitary in the physical Hilbert space of a finite dimension.

Additionally, we define the terminology "logical identity operator" as the logical operator that is equivalent to the identity operator on the code space.
From \Cref{thm:IdentityCondition} in \Cref{sec:IdentityCondition}, we see that a logical operator $\mathcal{O}$ becomes a logical identity operator if and only if it commutes with all the logical $Z$ and $X$ operators:
\begin{align}
    \mathcal{O} \sim I \iff
    {}^\forall\,\gamma \in H_1(C), \comm{\mathcal{O}}{\overline{Z}(\gamma)} = 0 \text{ and }
    {}^\forall\,\tilde{\gamma} \in H^1(C), \comm{\mathcal{O}}{\overline{X}(\tilde{\gamma})} = 0.
\end{align}

Here, we give a useful theorem that will be used in the following subsections, which is related to controlled versions of logical identity operators.
The formal statement is given as follows:

\begin{theorem}
    \label{thm:ControlledLogicallyIdentityOperator}
    If $\mathcal{O}_A = \exp(iA)$ is a logical identity operator, the operator $\mathcal{O}_A(\gamma)$ defined as
    \begin{align}
    \mathcal{O}_A(\gamma) \coloneq \exp(i a(\gamma) A)
    \end{align}
    with any 1-chain $\gamma$ is also a logical identity operator.
\end{theorem}

Intuitively, $\mathcal{O}_A(e_j)$ gives a controlled version of $\mathcal{O}_A$ with the control on the $j$th qubit, and becomes logically identical.
As a result, $\mathcal{O}_A(\gamma)$ is a product of such controlled operators and thus also a logical identity operator.
The formal proof is given as follows:
\begin{proof}
$\mathcal{O}_A$ can be decomposed as 
\begin{align}
    \mathcal{O}_A = \prod_i \exp(i a(e_i) A)^{\gamma^i} = \prod_i \mathcal{O}_A(e_i)^{\gamma^i}.
\end{align}
So, it is sufficient to show that $\mathcal{O}_A(e_i)$ is a logical identity operator for any edge $e_i$.
Each $\mathcal{O}_A(e_i)$ is calculated as
\begin{align}
    \mathcal{O}_A(e_i) = \exp(i z_i A) = \sum_{n=0}^\infty \frac{(iA)^n}{n!} z_i^n = 1 - z_i + z_i \sum_{n=0}^\infty \frac{(iA)^n}{n!} = P_i^+ + P_i^- \mathcal{O}_A,
\end{align}
where $P_i^+ = (I + Z_i)/2 = 1-z_i $ and $P_i^- = (I - Z_i)/2 = z_i$ are the projectors onto the $\ket{0}$ and $\ket{1}$ states of the $i$th qubit, respectively.
We used the fact that $z_i^n = z_i$ for any positive integer $n$.
Since $\mathcal{O}_A$ is a logical identity operator by the assumption, $\mathcal{O}_A(e_i)$ is logically equivalent to the identity operator:
\begin{align}
    \mathcal{O}_A(e_i) \sim P_i^+ + P_i^- \cdot I = I.
\end{align}
Thus, we have shown that $\mathcal{O}_A(\gamma)$ is a logical identity operator.
\end{proof}

We give a corollary of \Cref{thm:ControlledLogicallyIdentityOperator} in the following, which states that a product of the multicontrolled version of a logical identity operator is also a logical identity operator.
Applying \Cref{thm:ControlledLogicallyIdentityOperator} repeatedly, we immediately obtain this corollary:
\begin{corollary}
    \label{cor:ControlledLogicallyIdentityOperator}
    If $\mathcal{O}_A = \exp(iA)$ is a logical identity operator, the operator 
    \begin{align}
    \mathcal{O}_A(\gamma_1, \gamma_2, \ldots, \gamma_m) \coloneq \exp(ia(\gamma_1) a(\gamma_2) \cdots a(\gamma_m) A)
    \end{align}
    with any 1-chains $\gamma_1, \gamma_2, \ldots, \gamma_m$ is also a logical identity operator.
\end{corollary}
These theorem and corollary will be used in the following subsections.

\subsection{S gate}
We start with the logical $S$ gate.
Our logical $S$ gate that acts on a single logical qubit $\gamma \in H_1(C)$ is defined as
\begin{align}
    \overline{S}(\gamma) = \exp(i\frac{\pi}{2}a(\gamma)^2).
\end{align}

We show here that the logical operation of this gate depends only on the homology class of $\gamma$; i.e. 
$\overline{S}(\gamma_1) \sim \overline{S}(\gamma_2)$ if $\gamma_1 = \gamma_2 + \del_2 f$ for some $f \in C_2$.
Using the fact that $a$ is a linear function, we can decompose $\overline{S}(\gamma_1)$ as
\begin{align}
    \overline{S}(\gamma_1) 
    & = \overline{S}(\gamma_2) \cdot \exp(i \pi a(\gamma_2) a(\del_2 f)) \cdot \exp(i\frac{\pi}{2} a(\del_2 f)^2),
\end{align}
Actually, we can show that both $\exp(i \pi a(\gamma_2)a(\del_2 f))$ and $\exp(i\frac{\pi}{2}a(\del_2 f)^2)$ are logical identity operators as follows.

For $\exp(i \pi a(\gamma_2)a(\del_2 f))$, this is shown to be logically trivial by \Cref{thm:ControlledLogicallyIdentityOperator}.
If we set $\gamma \to \gamma_2$ and $A = a(\del_2 f)$, the operator $\mathcal{O}_A = \exp(i\pi a(\del_2 f)) = S^Z(f)$ is a stabilizer, and so it behaves trivially on the code space.
Then, from \Cref{thm:ControlledLogicallyIdentityOperator}, we see that $\exp(i \pi a(\gamma_2)a(\del_2 f))$ is also a logical identity operator.

Next, let us show that $\overline{S}(\del_2 f)$ is a logical identity operator.
There are two ways to show this.
The first way is what was given in the main text, using the fact that $\overline{S}(\del_2 f)$ agrees with the operator given by
\begin{align}
    \overline{S}(\del_2 f) & = \exp(i \frac{\pi}{2} \frac{I-S^Z(f)}{2}).
\end{align}
The second way is to show that $\overline{S}(\del_2 f)$ commutes with all the stabilizers and all the logical operators from \Cref{thm:IdentityCondition}.
Since $\overline{S}(\del_2 f)$ is constructed only from $z_i$, it trivially commutes with all the $Z$-type stabilizers and logical $Z$ operators, so here we check only the commutativity with the $X$-type stabilizers and logical $X$ operators.
Their commutation relation is calculated as
\begin{align}
    \overline{X}(\gamma) \overline{S}(\del_2 f) \overline{X}(\gamma)
    &= \exp(i\frac{\pi}{2} (a(\del_2 f) -2a(\del_2 f \cap \gamma)+ \inp{\del_2 f}{\gamma})^2) \\
    &= \exp(i\frac{\pi}{2} (a(\del_2 f) + \inp{f}{\del_2^T\gamma})^2)\exp(-2i\pi \Delta''),
\end{align}
where we used the commutation relation derived in \Cref{app:commutation_relation} and defined $\Delta''$ as $\Delta'' \coloneqq a(\del_2 f \cap \gamma)(a(\del_2 f) + \inp{f}{\del_2^T\gamma})$.
This $\Delta''$ is a diagonal matrix with integer eigenvalues in the computational basis, so we have $\exp(-2i\pi \Delta'') = I$.
 Regarding the term $\inp{f}{\del_2^T \gamma}$, it vanishes for both cases of $\gamma$ as follows:
\begin{align}
    \del_2^T \gamma &=
    \begin{cases}
        \del_2^T \del_1^T v & (v \in C_0) \text{ for } X\text{-type stabilizers} \\
        \del_2^T \tilde{\gamma} & (\tilde{\gamma} \in H^1(C)=\ker \del_2^T/\Im \del_1^T) \text{ for logical } X \text{ operator}
    \end{cases}
    \\
    & = 0.
\end{align}
Therefore, we find that the right-hand side reduces to $\overline{S}(\del_2 f)$:
\begin{align}
    \overline{X}(\gamma) \overline{S}(\del_2 f) \overline{X}(\gamma) 
    = \exp(i\frac{\pi}{2} a(\del_2 f)^2) 
    = \overline{S}(\del_2 f),
\end{align}
which shows that $\overline{S}(\del_2 f)$ is a logical identity operator.

Putting these results together, we have
\begin{align}
    \overline{S}(\gamma_1) 
    & \sim \exp(i\frac{\pi}{2}a(\gamma_2)^2) \cdot I \cdot I  = \overline{S}(\gamma_2),
\end{align}
which means that the logical operation of our logical $S$ gate depends only on the homology class of its argument.

\subsection{Hadamard gate}
Here, we show that the logical Hadamard gate construed in the main text
\begin{align}
    \overline{H}(\gamma) = e^{-i\frac{\pi}{4}} \exp(\frac{i\pi}{2}a(\gamma)^2) \exp(\frac{i\pi}{2}b(\tilde{\gamma})^2) \exp(\frac{i\pi}{2}a(\gamma)^2)
\end{align}
depends only on the homology class of $\gamma$, i.e. 
\begin{align}
    \label{eq:logical_hadamard_gate_cohomology}
    \overline{H}(\gamma_1) \sim  \overline{H}(\gamma_2)
\end{align}
for any $\gamma_1, \gamma_2 \in \ker{\del_1}(C)$ satisfying $\gamma_1 = \gamma_2 + \del_2 f$ with some $f \in C_2$.

In parallel to the proof of the logical $S$ gate, we can show that $\exp(i \pi b(\tilde{\gamma}_1)^2)$ is logically equivalent to $\exp(i\pi b(\tilde{\gamma}_2)^2/2)$.
Since $\tilde{\gamma}_1$ and $\tilde{\gamma}$ are elements of $\ker{\del_2^T}$, we can decompose $\tilde{\gamma}_1$ with some vertex $v$ as $\tilde{\gamma}_1 = \tilde{\gamma}_2 + \del_1^T v$.
Let us substitute this into $\exp(i\pi b(\tilde{\gamma}_1)^2/2)$, which gives
\begin{align}
    \exp(\frac{i\pi}{2}b(\tilde{\gamma}_1)^2) 
    = \exp(\frac{i\pi}{2}b(\tilde{\gamma}_2)^2) \exp(\frac{i\pi}{2}b(\del_1^T v)^2) \exp(i\pi b(\tilde{\gamma}_2) b(\del_1^T v)).
\end{align}
We can show that the second and third terms are logical identity operators, in a similar way to the electric gauge field case.
As for the second term, it agrees with the operator given by
\begin{align}
    \exp(\frac{i\pi}{2}b(\del_1^T v)^2) = \exp(i \frac{\pi}{2} \frac{I - S^X(v)}{2}),
\end{align}
which is clearly a logical identity operator since $S^X(v)$ is a stabilizer.
We can show that the third term is a logical identity operator, taking into account that $\exp(i\pi b(e_i) b(\del_1^T v))$ is a product of controlled-$X$ stabilizer gates with the vertex $v$ in the $+$ basis.
This is an example of the magnetic gauge field version of \Cref{thm:ControlledLogicallyIdentityOperator}.
Therefore, we find that 
\begin{align}
    \exp(\frac{i\pi}{2}b(\tilde{\gamma}_1)^2) \sim \exp(\frac{i\pi}{2}b(\tilde{\gamma}_2)^2).
\end{align}

Putting all things together, we find that the logical Hadamard gate with $\gamma_1$ is logically equivalent to that with $\gamma_2$.

\subsection{(Multi)controlled-Z gate}

Here, we show that the logical operation of the logical CZ gate constructed in the main text
\begin{align}
    \overline{CZ}(\gamma_1, \gamma_2) = \exp(i \pi a(\gamma_1) a(\gamma_2))
\end{align}
depends only on the homology class of its arguments $\gamma_1$ and $\gamma_2$, i.e.
\begin{align}
    \label{eq:logical_cz_gate_cohomology}
    \overline{CZ}(\gamma_1, \gamma_2) \sim \overline{CZ}(\gamma'_1, \gamma'_2)
\end{align}
for any $\gamma_1, \gamma_2, \gamma'_1, \gamma'_2 \in \ker{\del_1}(C)$ satisfying $\gamma_i = \gamma'_i + \del_2 f_i$ with some $f_i \in C_2$ for each $i=1,2$.

This can be shown straightforwardly as follows.
We find that $\overline{CZ}(\gamma_1, \gamma_2)$ can be decomposed with $\gamma'_1$ and $\gamma'_2$ as
\begin{multline}
    \exp(i\pi a(\gamma_1)a(\gamma_2)) = \exp(i\pi a(\gamma'_1) a(\gamma'_2))
    \\
    \times 
    \exp(i \pi a(\gamma'_1) a(\del_2 f_2))
    \exp(i \pi a(\del_2 f_1) a(\gamma'_2))
    \exp(i \pi a(\del_2 f_1)a(\del_2 f_2)).
\end{multline}
From \Cref{cor:ControlledLogicallyIdentityOperator}, we find that all the terms in the second line are logical identity operators since $\exp(i\pi a(\del_2 f_i))$ gives a stabilizer $S^Z(f_i)$ for each $i=1,2$.
Then, we can deduce the desired result in \Cref{eq:logical_cz_gate_cohomology}.

Similarly, we can show the logical operation of the logical $C^{m-1}Z$ gate
\begin{align}
    \overline{C^{m-1}Z}(\gamma_1, \ldots, \gamma_m) = \exp(i\pi a(\gamma_1) a(\gamma_2) \cdots a(\gamma_m))
\end{align}
depends only on the homology class of its arguments $\gamma_1, \ldots, \gamma_m$, i.e.
\begin{align}
    \label{eq:logical_cmz_gate_cohomology}
    \overline{C^{m-1}Z}(\gamma_1, \ldots, \gamma_m) \sim \overline{C^{m-1}Z}(\gamma'_1, \ldots, \gamma'_m),
\end{align}
where $\gamma_i, \gamma'_i \in \ker{\del_1}(C)$ satisfies $\gamma_i = \gamma'_i + \del_2 f_i$ with some $f_i \in C_2$ for each $i=1,\ldots,m$.

Its proof is similar to the logical CZ gate case.
We can decompose $\overline{C^{m-1}Z}(\gamma_1, \ldots, \gamma_m)$ as 
\begin{multline}
    \exp(i\pi a(\gamma_1) a(\gamma_2) \cdots a(\gamma_m)) = \exp(i\pi a(\gamma'_1) a(\gamma_2) \cdots a(\gamma_{m-1}) a(\gamma_m))
    \\ \times \exp(i\pi a(\del_2 f_1) a(\gamma_2) \cdots a(\gamma_{m-1}) a(\gamma_m)).
\end{multline}
Here, the second term is a logical identity operator from \Cref{cor:ControlledLogicallyIdentityOperator}, since $\exp(i\pi a(\del_2 f_1)) = S^Z(f_1)$ is a $Z$-type stabilizer.
Thus, we have 
\begin{align}
    \exp(i\pi a(\gamma_1) a(\gamma_2) \cdots a(\gamma_m)) \sim \exp(i\pi a(\gamma'_1) a(\gamma_2) \cdots a(\gamma_{m-1}) a(\gamma_m)).
\end{align}
Repetition of this argument leads to the result that for any $ 1 \leq k \leq m$,
\begin{align}
    \exp(i\pi a(\gamma_1) a(\gamma_2) \cdots a(\gamma_m)) \sim 
    \exp(i\pi a(\gamma'_1) a(\gamma'_2) \cdots a(\gamma'_k) a(\gamma_{k+1}) \cdots a(\gamma_m)).
\end{align}
Finally, taking $k=m$, we arrive at the desired result in \Cref{eq:logical_cmz_gate_cohomology}.

\subsection{T gate}
In this subsection, we show that the logical operation of the logical $T$ gate, 
\begin{align}
    \overline{T}(\gamma) = \exp(i\pi \qty(\frac{1}{2}a(\gamma)^3 - \frac{3}{4} a(\gamma)^2 + \frac{1}{2}a(\gamma))),
\end{align}
depends only on the homology class of its argument $\gamma$.
That is, we show
\begin{align}
    \label{eq:logical_t_gate_cohomology}
    \overline{T}(\gamma_1) \sim \overline{T}(\gamma_2)
\end{align}
for any $\gamma_1, \gamma_2 \in \ker{\del_1}(C)$ satisfying $\gamma_1 = \gamma_2 + \del_2 f$ with some $f \in C_2$.

This can be shown as follows.
Substituting $\gamma_1 = \gamma_2 + \del_2 f$ into the definition of the logical $T$ gate, we find that $\overline{T}(\gamma_1)$ can be decomposed as
\begin{align}
    \overline{T}(\gamma_1) = &\overline{T}(\gamma_2) \times \notag \\
     & \overline{T}(\del_2 f) \exp(i\frac{3\pi}{2}a(\gamma_2) a(\del_2 f) (a(\gamma_2) -1)) \exp(i\frac{3\pi}{2} a(\gamma_2)a(\del_2 f)^2).
\end{align}
Here, we can show that all the terms in the second line are logical identity operators as follows.
First, $\overline{T}(\del_2 f)$ agrees with the operator given by
\begin{align}
    \overline{T}(\del_2 f) = \exp(i \frac{\pi}{4} \frac{I - S^Z(f)}{2}),
\end{align}
which is clearly a logical identity operator since $S^Z(f)$ is a stabilizer.
Alternatively, we can show that $\overline{T}(\del_2 f)$ commutes with all the stabilizers and logical operators from \Cref{thm:IdentityCondition}, in a similar way to the logical $S$ gate case.
Next, let us consider the second term, which can be decomposed as
\begin{align}
    \exp(i\frac{3\pi}{2}a(\gamma_2) a(\del_2 f) (a(\gamma_2) -1)) 
    & = \prod_{i_1 < i_2 } \exp(3i\pi a(e_{i_1}) a(e_{i_2}) a(\del_2 f)),
\end{align}
where we denoted the support of $\gamma_2$ as $\{i_k\} \coloneqq \supp(\gamma_2)$, and used \Cref{eq:useful_formula_a_gamma_squared}.
Here, each factor $\exp(3i\pi a(e_{i_1}) a(e_{i_2}) a(\del_2 f))$ is a logical identity operator from \Cref{cor:ControlledLogicallyIdentityOperator}, since $\exp(3i\pi a(\del_2 f)) = S^Z(f)$ is a stabilizer.
Finally, we consider the last term, which can be written as
\begin{align}
    \exp(i\frac{3\pi}{2} a(\gamma_2)a(\del_2 f)^2) 
     = \exp(i a(\gamma_2) \frac{3\pi}{2}a(\del_2 f)^2).
\end{align}
Here, $\exp(3i\pi a(\del_2 f)^2/2) = \exp(-i\pi a(\del_2 f)^2/2) = \overline{S}^\dagger(\del_2 f)$ is the inverse of the logical $S$ gate with $\del_2 f$, which is a logical identity operator, as we have shown in the logical $S$ gate section.
Therefore, we find that the last term is also a logical identity operator from \Cref{thm:ControlledLogicallyIdentityOperator}.

Putting these results together, we have
\begin{align}
    \overline{T}(\gamma_1) \sim \overline{T}(\gamma_2),
\end{align}
which means that the logical $T$ gate depends only on the homology class of its argument on the code space.

\bibliographystyle{JHEP}
\bibliography{ref}

\providecommand{\href}[2]{#2}\begingroup\raggedright\begin{thebibliography}{10}

\bibitem{Shor:1994jg}
P.W.~Shor, \emph{{Polynomial time algorithms for prime factorization and
  discrete logarithms on a quantum computer}},
  \href{https://doi.org/10.1137/S0097539795293172}{\emph{SIAM J. Sci. Statist.
  Comput.} {\bfseries 26} (1997) 1484}
  [\href{https://arxiv.org/abs/quant-ph/9508027}{{\ttfamily
  quant-ph/9508027}}].

\bibitem{Grover:1996rk}
L.K.~Grover, \emph{{A fast quantum mechanical algorithm for database search}},
  in \emph{{Proceedings of the 28th Annual ACM Symposium on Theory of Computing
  (STOC 1996)}}, pp.~212--219, 1996,
  \href{https://doi.org/10.1145/237814.237866}{DOI}
  [\href{https://arxiv.org/abs/quant-ph/9605043}{{\ttfamily
  quant-ph/9605043}}].

\bibitem{Shor:1995hbe}
P.W.~Shor, \emph{{Scheme for reducing decoherence in quantum computer memory}},
  \href{https://doi.org/10.1103/physreva.52.r2493}{\emph{Phys. Rev. A}
  {\bfseries 52} (1995) R2493}.

\bibitem{Kitaev:1997wr}
A.Y.~Kitaev, \emph{{Fault tolerant quantum computation by anyons}},
  \href{https://doi.org/10.1016/S0003-4916(02)00018-0}{\emph{Annals Phys.}
  {\bfseries 303} (2003) 2}
  [\href{https://arxiv.org/abs/quant-ph/9707021}{{\ttfamily
  quant-ph/9707021}}].

\bibitem{Bombin:2006cd}
H.~Bombin and M.A.~Martin-Delgado, \emph{{Homological error correction:
  Classical and quantum codes}},
  \href{https://doi.org/10.1063/1.2731356}{\emph{J. Math. Phys.} {\bfseries 48}
  (2007) 052105} [\href{https://arxiv.org/abs/quant-ph/0605094}{{\ttfamily
  quant-ph/0605094}}].

\bibitem{Dennis_2002}
E.~Dennis, A.~Kitaev, A.~Landahl and J.~Preskill, \emph{Topological quantum
  memory}, \href{https://doi.org/10.1063/1.1499754}{\emph{Journal of
  Mathematical Physics} {\bfseries 43} (2002) 4452–4505}.

\bibitem{barends2014superconducting}
R.~Barends, J.~Kelly, A.~Megrant, A.~Veitia, D.~Sank, E.~Jeffrey et~al.,
  \emph{Superconducting quantum circuits at the surface code threshold for
  fault tolerance}, {\emph{Nature} {\bfseries 508} (2014) 500}.

\bibitem{Takita_2017}
M.~Takita, A.W.~Cross, A.~Córcoles, J.M.~Chow and J.M.~Gambetta,
  \emph{Experimental demonstration of fault-tolerant state preparation with
  superconducting qubits},
  \href{https://doi.org/10.1103/physrevlett.119.180501}{\emph{Physical Review
  Letters} {\bfseries 119} (2017) }.

\bibitem{acharya2024quantum}
R.~Acharya et~al., \emph{{Quantum error correction below the surface code
  threshold}}, \href{https://doi.org/10.1038/s41586-024-08449-y}{\emph{Nature}
  {\bfseries 638} (2025) 920}
  [\href{https://arxiv.org/abs/2408.13687}{{\ttfamily 2408.13687}}].

\bibitem{Wineland:1997mg}
D.J.~Wineland, C.~Monroe, W.M.~Itano, D.~Leibfried, B.E.~King and D.M.~Meekhof,
  \emph{{Experimental issues in coherent quantum-state manipulation of trapped
  atomic ions}}, \href{https://doi.org/10.6028/jres.103.019}{\emph{J. Res.
  Natl. Inst. Stand. Tech.} {\bfseries 103} (1998) 259}
  [\href{https://arxiv.org/abs/quant-ph/9710025}{{\ttfamily
  quant-ph/9710025}}].

\bibitem{Kielpinski:2002wbd}
D.~Kielpinski, C.~Monroe and D.J.~Wineland, \emph{{Architecture for a
  large-scale ion-trap quantum computer}},
  \href{https://doi.org/10.1038/nature00784}{\emph{Nature} {\bfseries 417}
  (2002) 709}.

\bibitem{Ransford:2025ksn}
A.~Ransford et~al., \emph{{Helios: A 98-qubit trapped-ion quantum computer}},
  \href{https://arxiv.org/abs/2511.05465}{{\ttfamily 2511.05465}}.

\bibitem{Bluvstein:2023zmt}
D.~Bluvstein et~al., \emph{{Logical quantum processor based on reconfigurable
  atom arrays}},
  \href{https://doi.org/10.1038/s41586-023-06927-3}{\emph{Nature} {\bfseries
  626} (2024) 58} [\href{https://arxiv.org/abs/2312.03982}{{\ttfamily
  2312.03982}}].

\bibitem{Rodriguez:2024bhh}
P.S.~Rodriguez et~al., \emph{{Experimental demonstration of logical magic state
  distillation}},
  \href{https://doi.org/10.1038/s41586-025-09367-3}{\emph{Nature} {\bfseries
  645} (2025) 620} [\href{https://arxiv.org/abs/2412.15165}{{\ttfamily
  2412.15165}}].

\bibitem{Zhong:2020bnd}
T.~Zhong et~al., \emph{{Quantum computational advantage using photons}},
  \href{https://doi.org/10.1126/science.abe8770}{\emph{Science} {\bfseries 370}
  (2020) 1460} [\href{https://arxiv.org/abs/2005.03850}{{\ttfamily
  2005.03850}}].

\bibitem{Alexander:2024vys}
K.~Alexander et~al., \emph{{A manufacturable platform for photonic quantum
  computing}}, \href{https://doi.org/10.1038/s41586-025-08820-7}{\emph{Nature}
  {\bfseries 641} (2025) 876}
  [\href{https://arxiv.org/abs/2404.17570}{{\ttfamily 2404.17570}}].

\bibitem{poulin2008iterativedecodingsparsequantum}
D.~Poulin and Y.~Chung, \emph{{On the iterative decoding of sparse quantum
  codes}}, \href{https://doi.org/10.26421/QIC8.10-8}{\emph{Quantum Inf.
  Comput.} {\bfseries 8} (2008) 987}
  [\href{https://arxiv.org/abs/0801.1241}{{\ttfamily 0801.1241}}].

\bibitem{Tillich_2014}
J.-P.~Tillich and G.~Zemor, \emph{Quantum ldpc codes with positive rate and
  minimum distance proportional to the square root of the blocklength},
  \href{https://doi.org/10.1109/tit.2013.2292061}{\emph{IEEE Transactions on
  Information Theory} {\bfseries 60} (2014) 1193–1202}.

\bibitem{Panteleev:2019upy}
P.~Panteleev and G.~Kalachev, \emph{{Degenerate Quantum LDPC Codes With Good
  Finite Length Performance}},
  \href{https://doi.org/10.22331/q-2021-11-22-585}{\emph{Quantum} {\bfseries 5}
  (2021) 585} [\href{https://arxiv.org/abs/1904.02703}{{\ttfamily
  1904.02703}}].

\bibitem{Zeng:2007uho}
B.~Zeng, A.~Cross and I.L.~Chuang, \emph{{Transversality versus Universality
  for Additive Quantum Codes}},
  \href{https://doi.org/10.1109/TIT.2011.2161917}{\emph{IEEE Trans. Info.
  Theor.} {\bfseries 57} (2011) 6272}
  [\href{https://arxiv.org/abs/0706.1382}{{\ttfamily 0706.1382}}].

\bibitem{Eastin:2009tem}
B.~Eastin and E.~Knill, \emph{{Restrictions on Transversal Encoded Quantum Gate
  Sets}}, \href{https://doi.org/10.1103/PhysRevLett.102.110502}{\emph{Phys.
  Rev. Lett.} {\bfseries 102} (2009) 110502}
  [\href{https://arxiv.org/abs/0811.4262}{{\ttfamily 0811.4262}}].

\bibitem{Bravyi:2012rnv}
S.~Bravyi and R.~Koenig, \emph{{Classification of Topologically Protected Gates
  for Local Stabilizer Codes}},
  \href{https://doi.org/10.1103/physrevlett.110.170503}{\emph{Phys. Rev. Lett.}
  {\bfseries 110} (2013) 170503}
  [\href{https://arxiv.org/abs/1206.1609}{{\ttfamily 1206.1609}}].

\bibitem{Horsman_2012}
D.~Horsman, A.G.~Fowler, S.~Devitt and R.V.~Meter, \emph{Surface code quantum
  computing by lattice surgery},
  \href{https://doi.org/10.1088/1367-2630/14/12/123011}{\emph{New Journal of
  Physics} {\bfseries 14} (2012) 123011}.

\bibitem{Fowler_2012}
A.G.~Fowler, M.~Mariantoni, J.M.~Martinis and A.N.~Cleland, \emph{Surface
  codes: Towards practical large-scale quantum computation},
  \href{https://doi.org/10.1103/physreva.86.032324}{\emph{Physical Review A}
  {\bfseries 86} (2012) }.

\bibitem{Litinski_2019}
D.~Litinski, \emph{A game of surface codes: Large-scale quantum computing with
  lattice surgery},
  \href{https://doi.org/10.22331/q-2019-03-05-128}{\emph{Quantum} {\bfseries 3}
  (2019) 128}.

\bibitem{Bombin:2009tbs}
H.~Bombin and M.A.~Martin-Delgado, \emph{{Quantum measurements and gates by
  code deformation}},
  \href{https://doi.org/10.1088/1751-8113/42/9/095302}{\emph{J. Phys. A}
  {\bfseries 42} (2009) 095302}
  [\href{https://arxiv.org/abs/0704.2540}{{\ttfamily 0704.2540}}].

\bibitem{Knill:2004ctr}
E.~Knill, \emph{{Fault-Tolerant Postselected Quantum Computation: Schemes}},
  \href{https://arxiv.org/abs/quant-ph/0402171}{{\ttfamily quant-ph/0402171}}.

\bibitem{Bravyi:2004isx}
S.~Bravyi and A.~Kitaev, \emph{{Universal quantum computation with ideal
  Clifford gates and noisy ancillas}},
  \href{https://doi.org/10.1103/PhysRevA.71.022316}{\emph{Phys. Rev. A}
  {\bfseries 71} (2005) 022316}
  [\href{https://arxiv.org/abs/quant-ph/0403025}{{\ttfamily
  quant-ph/0403025}}].

\bibitem{Barkeshli:2022edm}
M.~Barkeshli, Y.-A.~Chen, P.-S.~Hsin and R.~Kobayashi, \emph{{Higher-group
  symmetry in finite gauge theory and stabilizer codes}},
  \href{https://doi.org/10.21468/SciPostPhys.16.4.089}{\emph{SciPost Phys.}
  {\bfseries 16} (2024) 089}
  [\href{https://arxiv.org/abs/2211.11764}{{\ttfamily 2211.11764}}].

\bibitem{Zhu:2023xfg}
G.~Zhu, S.~Sikander, E.~Portnoy, A.W.~Cross and B.J.~Brown, \emph{{Non-Clifford
  and parallelizable fault-tolerant logical gates on constant and
  almost-constant rate homological quantum LDPC codes via higher symmetries}},
  \href{https://doi.org/10.1103/wcxs-w69t}{\emph{PRX Quantum} {\bfseries 6}
  (2025) 040361} [\href{https://arxiv.org/abs/2310.16982}{{\ttfamily
  2310.16982}}].

\bibitem{Lin:2024uhb}
T.-C.~Lin, \emph{{Transversal non-Clifford gates for quantum LDPC codes on
  sheaves}},  \href{https://arxiv.org/abs/2410.14631}{{\ttfamily 2410.14631}}.

\bibitem{Golowich:2024ogv}
L.~Golowich and T.-C.~Lin, \emph{{Quantum LDPC Codes with Transversal
  Non-Clifford Gates via Products of Algebraic Codes}},  in \emph{{57th Annual
  ACM Symposium on Theory of Computing}}, 10, 2024,
  \href{https://doi.org/10.1145/3717823.3718139}{DOI}
  [\href{https://arxiv.org/abs/2410.14662}{{\ttfamily 2410.14662}}].

\bibitem{breuckmann2024cupsgatesicohomology}
N.P.~Breuckmann, M.~Davydova, J.N.~Eberhardt and N.~Tantivasadakarn,
  \emph{{Cups and Gates I: Cohomology invariants and logical quantum
  operations}}, \href{https://doi.org/10.1007/s00220-026-05570-z}{\emph{Commun.
  Math. Phys.} {\bfseries 407} (2026) 86}
  [\href{https://arxiv.org/abs/2410.16250}{{\ttfamily 2410.16250}}].

\bibitem{hsin2024classifyinglogicalgatesquantum}
P.-S.~Hsin, R.~Kobayashi and G.~Zhu, \emph{{Classifying Logical Gates in
  Quantum Codes via Cohomology Operations and Symmetry}},
  \href{https://arxiv.org/abs/2411.15848}{{\ttfamily 2411.15848}}.

\bibitem{zhu2025topologicaltheoryqldpcnonclifford}
G.~Zhu, \emph{{A topological theory for qLDPC: non-Clifford gates and magic
  state fountain on homological product codes with constant rate and beyond the
  $N^{1/3}$ distance barrier}},
  \href{https://arxiv.org/abs/2501.19375}{{\ttfamily 2501.19375}}.

\bibitem{Kobayashi:2025cfh}
R.~Kobayashi, G.~Zhu and P.-S.~Hsin, \emph{{Clifford Hierarchy Stabilizer
  Codes: Transversal Non-Clifford Gates and Magic}},
  \href{https://doi.org/10.1103/j2q4-bnc6}{\emph{Phys. Rev. Lett.} {\bfseries
  136} (2026) 250802} [\href{https://arxiv.org/abs/2511.02900}{{\ttfamily
  2511.02900}}].

\bibitem{Calderbank_1996}
A.R.~Calderbank and P.W.~Shor, \emph{Good quantum error-correcting codes
  exist}, \href{https://doi.org/10.1103/physreva.54.1098}{\emph{Physical Review
  A} {\bfseries 54} (1996) 1098–1105}.

\bibitem{Steane:1996ghp}
A.M.~Steane, \emph{{Error Correcting Codes in Quantum Theory}},
  \href{https://doi.org/10.1103/physrevlett.77.793}{\emph{Phys. Rev. Lett.}
  {\bfseries 77} (1996) 793}.

\bibitem{Kitaev1997quantum}
A.Y.~Kitaev, \emph{Quantum computations: algorithms and error correction},
  {\emph{Russian Mathematical Surveys} {\bfseries 52} (1997) 1191}.

\bibitem{Kogut:1974ag}
J.B.~Kogut and L.~Susskind, \emph{{Hamiltonian Formulation of Wilson's Lattice
  Gauge Theories}}, \href{https://doi.org/10.1103/PhysRevD.11.395}{\emph{Phys.
  Rev. D} {\bfseries 11} (1975) 395}.

\bibitem{Kogut:1979wt}
J.B.~Kogut, \emph{{An Introduction to Lattice Gauge Theory and Spin Systems}},
  \href{https://doi.org/10.1103/RevModPhys.51.659}{\emph{Rev. Mod. Phys.}
  {\bfseries 51} (1979) 659}.

\bibitem{Gottesman:1997zz}
D.~Gottesman, \emph{{Stabilizer codes and quantum error correction}}, Ph.D.
  thesis, California Institute of Technology, 1997.
\newblock \href{https://arxiv.org/abs/quant-ph/9705052}{{\ttfamily
  quant-ph/9705052}}.

\bibitem{Panteleev:2021wvc}
P.~Panteleev and G.~Kalachev, \emph{{Asymptotically good Quantum and locally
  testable classical LDPC codes}},  in \emph{{54th Annual ACM Symposium on
  Theory of Computing}}, 11, 2021,
  \href{https://doi.org/10.1145/3519935.3520017}{DOI}
  [\href{https://arxiv.org/abs/2111.03654}{{\ttfamily 2111.03654}}].

\bibitem{Hastings:2020jbo}
M.B.~Hastings, J.~Haah and R.~O'Donnell, \emph{{Fiber bundle codes: breaking
  the $n^{1/2}\,\mathrm{polylog}(n)$ barrier for Quantum LDPC codes}},  in
  \emph{{Proceedings of the 53rd Annual ACM SIGACT Symposium on Theory of
  Computing (STOC 2021)}}, pp.~1276--1288, 2021,
  \href{https://doi.org/10.1145/3406325.3451005}{DOI}
  [\href{https://arxiv.org/abs/2009.03921}{{\ttfamily 2009.03921}}].

\bibitem{Old:2025lsb}
J.~Old, J.~Bechar, M.~M{\"u}ller and S.~Heu{\ss}en, \emph{{Addressable
  fault-tolerant universal quantum gate operations for high-rate lift-connected
  surface codes}},  \href{https://arxiv.org/abs/2511.10191}{{\ttfamily
  2511.10191}}.

\bibitem{Yoder:2016bys}
T.J.~Yoder, R.~Takagi and I.L.~Chuang, \emph{{Universal Fault-Tolerant Gates on
  Concatenated Stabilizer Codes}},
  \href{https://doi.org/10.1103/PhysRevX.6.031039}{\emph{Phys. Rev. X}
  {\bfseries 6} (2016) 031039}
  [\href{https://arxiv.org/abs/1603.03948}{{\ttfamily 1603.03948}}].

\bibitem{Tansuwannont:2025riy}
T.~Tansuwannont, Y.~Takada and K.~Fujii, \emph{{Clifford gates with logical
  transversality for self-dual CSS codes}},
  \href{https://arxiv.org/abs/2503.19790}{{\ttfamily 2503.19790}}.

\bibitem{Anderson:2014voa}
J.T.~Anderson and T.~Jochym-O'Connor, \emph{{Classification of transversal
  gates in qubit stabilizer codes}},
  \href{https://doi.org/10.26421/QIC16.9-10-3}{\emph{Quant. Inf. Comput.}
  {\bfseries 16} (2016) 0771}
  [\href{https://arxiv.org/abs/1409.8320}{{\ttfamily 1409.8320}}].

\bibitem{Gottesman:1999tea}
D.~Gottesman and I.L.~Chuang, \emph{{Demonstrating the viability of universal
  quantum computation using teleportation and single-qubit operations}},
  \href{https://doi.org/10.1038/46503}{\emph{Nature} {\bfseries 402} (1999)
  390} [\href{https://arxiv.org/abs/quant-ph/9908010}{{\ttfamily
  quant-ph/9908010}}].

\bibitem{Kribs:2004dqj}
D.~Kribs, R.~Laflamme and D.~Poulin, \emph{{Unified and Generalized Approach to
  Quantum Error Correction}},
  \href{https://doi.org/10.1103/PhysRevLett.94.180501}{\emph{Phys. Rev. Lett.}
  {\bfseries 94} (2005) 180501}
  [\href{https://arxiv.org/abs/quant-ph/0412076}{{\ttfamily
  quant-ph/0412076}}].

\bibitem{Kribs:2006doz}
D.W.~Kribs, R.~Laflamme, D.~Poulin and M.~Lesosky, \emph{{Operator quantum
  error correction}}, \href{https://doi.org/10.26421/QIC6.4-5-6}{\emph{Quant.
  Inf. Comput.} {\bfseries 6} (2006) 382}
  [\href{https://arxiv.org/abs/quant-ph/0504189}{{\ttfamily
  quant-ph/0504189}}].

\bibitem{Campbell:2019cyo}
E.T.~Campbell, \emph{{A theory of single-shot error correction for adversarial
  noise}}, \href{https://doi.org/10.1088/2058-9565/aafc8f}{\emph{Quantum Sci.
  Technol.} {\bfseries 4} (2019) 025006}
  [\href{https://arxiv.org/abs/1805.09271}{{\ttfamily 1805.09271}}].

\bibitem{Hastings:2016jmy}
M.B.~Hastings, \emph{{Quantum codes from high-dimensional manifolds}},  in
  \emph{{8th Innovations in Theoretical Computer Science Conference (ITCS
  2017)}}, vol.~67 of \emph{LIPIcs}, pp.~25:1--25:26, 2017,
  \href{https://doi.org/10.4230/LIPIcs.ITCS.2017.25}{DOI}
  [\href{https://arxiv.org/abs/1608.05089}{{\ttfamily 1608.05089}}].

\bibitem{Creutz:1982dn}
M.~Creutz and M.~Okawa, \emph{{Generalized Actions in $Z(p$) Lattice Gauge
  Theory}}, \href{https://doi.org/10.1016/0550-3213(83)90220-1}{\emph{Nucl.
  Phys. B} {\bfseries 220} (1983) 149}.

\bibitem{Creutz:1983ev}
M.~Creutz, L.~Jacobs and C.~Rebbi, \emph{{Monte Carlo Computations in Lattice
  Gauge Theories}},
  \href{https://doi.org/10.1016/0370-1573(83)90016-9}{\emph{Phys. Rept.}
  {\bfseries 95} (1983) 201}.

\bibitem{Weinberg:1967tq}
S.~Weinberg, \emph{{A Model of Leptons}},
  \href{https://doi.org/10.1103/PhysRevLett.19.1264}{\emph{Phys. Rev. Lett.}
  {\bfseries 19} (1967) 1264}.

\bibitem{Petcher:1980cq}
D.~Petcher and D.H.~Weingarten, \emph{{Monte Carlo Calculations and a Model of
  the Phase Structure for Gauge Theories on Discrete Subgroups of SU(2)}},
  \href{https://doi.org/10.1103/PhysRevD.22.2465}{\emph{Phys. Rev. D}
  {\bfseries 22} (1980) 2465}.

\bibitem{Bhanot:1981xp}
G.~Bhanot and C.~Rebbi, \emph{{Monte Carlo Simulations of Lattice Models With
  Finite Subgroups of SU(3) as Gauge Groups}},
  \href{https://doi.org/10.1103/PhysRevD.24.3319}{\emph{Phys. Rev. D}
  {\bfseries 24} (1981) 3319}.

\bibitem{Bhanot:1981pj}
G.~Bhanot, \emph{{SU(3) Lattice Gauge Theory in Four-dimensions With a Modified
  Wilson Action}},
  \href{https://doi.org/10.1016/0370-2693(82)91207-2}{\emph{Phys. Lett. B}
  {\bfseries 108} (1982) 337}.

\bibitem{Kubischta:2023nlb}
E.~Kubischta and I.~Teixeira, \emph{{Family of Quantum Codes with Exotic
  Transversal Gates}},
  \href{https://doi.org/10.1103/PhysRevLett.131.240601}{\emph{Phys. Rev. Lett.}
  {\bfseries 131} (2023) 240601}
  [\href{https://arxiv.org/abs/2305.07023}{{\ttfamily 2305.07023}}].

\bibitem{Denys:2023syu}
A.~Denys and A.~Leverrier, \emph{{Quantum Error-Correcting Codes with a
  Covariant Encoding}},
  \href{https://doi.org/10.1103/PhysRevLett.133.240603}{\emph{Phys. Rev. Lett.}
  {\bfseries 133} (2024) 240603}
  [\href{https://arxiv.org/abs/2306.11621}{{\ttfamily 2306.11621}}].

\bibitem{Carena:2024dzu}
M.~Carena, H.~Lamm, Y.-Y.~Li and W.~Liu, \emph{{Quantum error thresholds for
  gauge-redundant digitizations of lattice field theories}},
  \href{https://doi.org/10.1103/PhysRevD.110.054516}{\emph{Phys. Rev. D}
  {\bfseries 110} (2024) 054516}
  [\href{https://arxiv.org/abs/2402.16780}{{\ttfamily 2402.16780}}].

\bibitem{Assi:2024pdn}
B.~Assi and H.~Lamm, \emph{{Digitization and subduction of SU(N) gauge
  theories}}, \href{https://doi.org/10.1103/PhysRevD.110.074511}{\emph{Phys.
  Rev. D} {\bfseries 110} (2024) 074511}
  [\href{https://arxiv.org/abs/2405.12204}{{\ttfamily 2405.12204}}].

\bibitem{Lin:2024utr}
C.-J.~Lin, Z.-W.~Liu, V.V.~Albert and A.V.~Gorshkov, \emph{{Covariant Quantum
  Error-Correcting Codes with Metrological Entanglement Advantage}},
  \href{https://doi.org/10.1103/dttc-ksdn}{\emph{Phys. Rev. Lett.} {\bfseries
  135} (2025) 110801} [\href{https://arxiv.org/abs/2409.20561}{{\ttfamily
  2409.20561}}].

\bibitem{Kurkcuoglu:2025gik}
D.M.~K{\"u}rk{\c{c}}{\"u}oglu, H.~Lamm, O.~Ogunkoya and L.~Pierattelli,
  \emph{{The 2T-quoctit: a two-mode bosonic qudit for high energy physics}},
  \href{https://arxiv.org/abs/2509.21941}{{\ttfamily 2509.21941}}.

\bibitem{Perez:2025cxl}
S.O.~Perez, E.M.~Murairi, E.J.~Gustafson and H.~Lamm, \emph{{Primitive Quantum
  Gates for an $SU(3)$ Discrete Subgroup: $\Sigma(72\times3)$}},
  \href{https://arxiv.org/abs/2511.17437}{{\ttfamily 2511.17437}}.

\end{thebibliography}\endgroup

\end{document}